\documentclass[pdflatex,sn-mathphys-num]{sn-jnl}

\pdfoutput=1

\usepackage{graphicx}%
\usepackage{multirow}%
\usepackage{amsmath,amssymb,amsfonts}%
\usepackage{amsthm}%
\usepackage{mathrsfs}%
\usepackage[title]{appendix}%
\usepackage{xcolor}%
\usepackage{textcomp}%
\usepackage{manyfoot}%
\usepackage{booktabs}%
\usepackage{algorithm}%
\usepackage{algorithmicx}%
\usepackage{algpseudocode}%
\usepackage{listings}%

\usepackage{graphicx}
\usepackage{placeins}
\usepackage{rotating}
\usepackage{fancyhdr}
\usepackage{placeins}
\usepackage{float}
\usepackage{booktabs}
\usepackage{multirow}
\usepackage{braket}
\usepackage{comment}
\usepackage{bbm}
\usepackage{nicefrac}
\usepackage{pifont}
\usepackage{tikz}
\usepackage{bookmark}
\usepackage[caption=false]{subfig}
\usepackage{xcolor}
\usepackage{ulem}

\usetikzlibrary{calc}

\theoremstyle{thmstyleone}%
\newtheorem{theorem}{Theorem}
\newtheorem{lemma}[theorem]{Lemma}
\newtheorem{definition}[theorem]{Definition}

\newcommand\norm[1]{\lVert#1\rVert}

\newcommand{\Id}{\text{Id}}
\newcommand{\R}{\mathbb{R}}
\newcommand{\N}{\mathbb{N}}
\newcommand{\Z}{\mathbb{Z}}
\newcommand{\C}{\mathbb{C}}

\newcommand{\lset}[1]{\left\{ #1 \right\}}
\newcommand{\End}[1]{\text{End}(#1)}
\newcommand{\B}{\mathcal{B}}
\newcommand{\BR}{\mathcal{B}^{\otimes R}}

\newcommand{\U}{\mathcal{U}}

\newcommand{\btheta}{\boldsymbol \theta}
\newcommand{\bomega}{\boldsymbol \omega}
\newcommand{\bOmega}{\boldsymbol \Omega}
\newcommand{\blambda}{\boldsymbol \lambda}

\newcommand{\bx}{\boldsymbol x}
\newcommand{\by}{\boldsymbol y}
\newcommand{\bj}{\boldsymbol j}
\newcommand{\bi}{\boldsymbol i}
\newcommand{\bk}{\boldsymbol k}

\newcommand{\diag}{\text{diag}}
\newcommand{\ord}{\text{ord}}

\newcommand{\cmark}{\ding{51}}
\newcommand{\xmark}{\ding{55}}

\begin{document}

\title[Spectral Invariance and Maximality Properties of the Frequency Spectrum of Quantum Neural Networks]{Spectral Invariance and Maximality Properties of the Frequency Spectrum of Quantum Neural Networks}


\author*[1]{\fnm{Patrick} \sur{Holzer}}\email{patrick.holzer@itwm.fraunhofer.de}
\equalcont{These authors contributed equally to this work.}

\author*[1]{\fnm{Ivica} \sur{Turkalj}}\email{ivica.turkalj@itwm.fraunhofer.de}
\equalcont{These authors contributed equally to this work.}

\affil*[1]{\orgdiv{Analytics and Computing}, \orgname{Fraunhofer Institute for Industrial Mathematics ITWM}, \orgaddress{\street{Fraunhofer-Platz 1}, \city{Kaiserslautern}, \postcode{67663}, \country{Germany}}}


\abstract{Quantum Neural Networks (QNNs) are a popular approach in Quantum Machine Learning.
We analyze the
frequency spectrum using Minkowski sums, which yields a compact algebraic description and permits explicit computation. Using this description, we prove several maximality results for broad classes of QNN architectures.
Under some mild technical conditions we establish a bijection between classes of models with the same area $A:=R\cdot L$ that preserves the frequency spectrum, where $R$ denotes the number of qubits and $L$ the number of layers, which we consequently call spectral invariance under area-preserving transformations.
With this we explain the symmetry in $R$ and $L$ in the results often observed in the literature and show that the maximal frequency spectrum depends only on the area $A=RL$ and not on the individual values of $R$ and $L$. 
Moreover, we collect and extend existing results and specify the maximum possible frequency spectrum of a QNN with an arbitrary number of layers as a function of the spectrum of its generators. In the case of arbitrary dimensional generators, where our two introduced notions of maximality differ, we extend existing Golomb ruler based results and introduce a second novel approach based on a variation of the turnpike problem, which we call the relaxed turnpike problem. 
We clarify comprehensively how the generators of a QNN must
be chosen in order to obtain a maximal frequency spectrum for a given area $A$, thereby contributing to a deeper theoretical understanding.
However, our numerical experiments show that trainability depends not only on $A = RL$, but also on the choice of $(R,L)$,
so that knowledge of the maximum frequency spectrum alone is not sufficient to ensure good trainability.
While the maximality analysis limits the space of attainable frequency spectra and characterizing these spectra helps to understand the effect of the ansatz choice, it does not by itself provide sufficient criteria for ansatz selection, and further considerations are required.
}




\maketitle

\section{Introduction}
\label{sec:chapter_intro}
A frequently studied approach to Quantum Machine Learning (QML), a field combining quantum computing and classical Machine Learning (ML), is based on Variational Quantum Algorithms (VQAs) \cite{schuld2021machine, Cerezo2020VariationalQA, Cerezo2022ChallengesAO}. 
VQAs are hybrid algorithms that use parameterised quantum circuits (PQCs) to build a target function. PQCs are quantum circuits $U(\btheta)$ dependent on some real parameters $\btheta \in \R^P$. While the evaluation of the target function is done on quantum hardware, the training or optimisation of the parameters $\btheta$ is done on classical hardware.
VQAs are considered as a promising candidate for practical applications on Noisy Intermediate-Scale Quantum (NISQ) computers \cite{Preskill2018QuantumCI, bharti2022noisy, Cerezo2020VariationalQA}.
Different types of architectures or ans\"atze for PQCs have been proposed, like the Quantum Alternating Operator Ansatz (QAOA) \cite{Hadfield2017FromTQ}, Variational Quantum Eigensolver (VQE) \cite{Tilly2021TheVQ, Kandala2017HardwareefficientVQ}, Quantum Neurons \cite{Kordzanganeh2022AnEF} and Quantum Neural Networks (QNNs), which appear under different names in the literature \cite{Schuld2020EffectOD, McClean2015TheTO, Romero2019VariationalQG, Mitarai2018QuantumCL, Farhi2018ClassificationWQ, McClean2018BarrenPI, Benedetti2019ParameterizedQC}. There are also variations and extension of the QNN ansatz like Dissipative Quantum Neural Networks (dQNNs) \cite{Beer2019TrainingDQ, Beer2021TrainingQN, Heimann2022LearningCO} or Hybrid Classical-Quantum Neural Networks (HQNNs) \cite{Vidal2019InputRF, Mitarai2018QuantumCL, Kashif2022DemonstrQA}. More hardware near approaches \cite{Gan2021FockSE} or higher dimensional qubits (qudits) are also investigated \cite{Casas2023MultidimensionalFS}. Additionally, there is no standardized definition in the literature for a QNN, similar to classical neural networks. Therefore, there are slight variations in the approaches used. The application of PQCs and QNNs in finance \cite{Wolf2023QAS, Ors2018QuantumCF, herman2022survey}, medicine \cite{Maheshwari2022QuantumML, Umer2022ACS}, chemistry \cite{Xia2022PotentialAO} and other domains has also been investigated.

Following \cite{Schuld2020EffectOD, Shin2022ExponentialDE, Kordzanganeh2022AnEF}, this work focuses solely on pure QNNs. 
A QNN is a function of the form 
\begin{align}
    f_{\btheta}(\bx) = \langle 0| U^\dagger(\bx,\btheta) M U(\bx,\btheta)|0\rangle
\end{align}
for some observable $M$, $\bx \in \R^N$ and parameter $\btheta\in \R^P$, where the parameterised quantum circuit $U(\bx,\btheta)$ consists of embedding layers of a certain form. The dimension of the underlying Hilbert space is denoted by $d:=2^R$, where the number of qubits is denoted by $R$. If $N=1$, then $x\in \R$ and $f:\R\rightarrow \R$ is a real valued \textbf{univariate function}. In this case, each layer consists of a data encoding circuit of the form $S_l(x)=e^{-ix H_l}$ for some Hamiltonian $H_l$ called \textbf{generator}, and a trainable parameter encoding circuit $W_{\btheta}^{(l)}$. Concretely, the parameterised quantum circuit takes the form
\begin{align}
    U(x,\btheta) = W^{(L+1)}_{\btheta}\underbrace{S_L(x)W^{(L)}_{\btheta}}_{\text{Layer } L}\cdots W^{(2)}_{\btheta} \underbrace{S_1(x)W^{(1)}_{\btheta}}_{\text{Layer } 1}
\end{align}
The reuse of the input $x$ is called \textbf{data re-uploading} or \textbf{input redundancy} and has been shown to be necessary to increase the expressiveness of the model \cite{PerezSalinas2019DataRF, Kordzanganeh2022AnEF, Vidal2019InputRF}.

The univariate model can be extended to a \textbf{multivariate model} with $\bx\in \R^N$ for $N>1$ in multiple ways, we present two typically used approaches called \textbf{sequential} and \textbf{parallel ansatz}.
In the sequential ansatz, there is a parametrized quantum circuit $U_n(x_n, \btheta)$ for each univariate variable $x_n$. The full parameterised circuit is then defined via concatenation
\begin{align}
    U(\bx, \btheta) := \prod_{n=1}^N U_n(x_n, \btheta).
\end{align}
In the parallel ansatz, the data encoding circuits $S_{n, l}(x_n)$ for each univariate variable $x_n$ and each layer $l$ are glued together in parallel via tensor product
\begin{align}
\label{eq: parallel ansatz}
     S_l(\bx) := \bigotimes_{n=1}^N S_{l, n}(x_n).
\end{align}
The circuits $ U(\bx, \btheta)$ of the sequential and parallel ansatz are shown in Figure \ref{fig:qnn circuits parallel and sequential}.
\begin{figure}
    \centering
    \subfloat[\centering Parallel ansatz]{\resizebox{0.4\linewidth}{!}{
        \begin{tikzpicture}[scale = 0.25, line width = 0.5]
    \definecolor{colorblue}{HTML}{007AFF}
    \definecolor{colorred}{HTML}{D20000}
    \definecolor{colorgreen}{HTML}{7EC636}
   
    \def \roundness{3};
    \def \roundnesssubgenerators{1};
    \def \widthunitary{2};
    \def \widthsubgenerator{1.3};
    \def \heightS{8};
    \def \distunitariesx{1};
    \def \distunitariesy{4};
    \def \heightW{\heightS + \distunitariesy + \heightS};

    \def \yupper{\heightS + \distunitariesy};

    \def \xtwo{2 *\widthunitary + 2 * \distunitariesx};
    \def \xthree{\xtwo + 3*\widthunitary + \distunitariesx};
    \def \xfour{\xthree + 2 *\widthunitary + 2 * \distunitariesx};

    \coordinate (origin) at (0, 0);
    \coordinate (S) at (\widthunitary, \heightS);
    \coordinate (W) at (\widthunitary, \heightW);
    \coordinate (subgenerator) at (\widthsubgenerator, \widthsubgenerator);
    
    \newcommand{\figureline}[1]{
    \draw (-2, #1) -- (23, #1);
    \node[left] at (-2, #1) {$|0\rangle$};
    }

    \newcommand{\figureW}[2]{
        \draw[rounded corners = \roundness, fill = colorred] (#1, #2) rectangle ++(W);
    }

    \newcommand{\figureS}[2]{
        \draw[rounded corners = \roundness, fill = colorblue] (#1, #2) rectangle ++(S);
        \draw[rounded corners = \roundnesssubgenerators, fill = white] (#1 + 0.5*\widthunitary-0.5*\widthsubgenerator, #2 + 0.5) rectangle ++(subgenerator);
        \draw[fill = white] (#1 + 0.5*\widthunitary-0.5*\widthsubgenerator, #2 + 6.2) rectangle ++(subgenerator);
        \draw[fill = white] (#1 + 0.5*\widthunitary-0.5*\widthsubgenerator, #2 + 4) rectangle ++(subgenerator);
        \draw[dotted] (#1 + 0.5 * \widthunitary, #2 + 2.5) -- ++(0, 1);
    }
    \newcommand{\figureLayerS}[2]{
        \figureS{#1}{#2}
        \figureS{#1}{#2 + \yupper} 
        \draw[dotted] (#1 + 0.5 * \widthunitary, #2 + \heightS + 1) -- ++(0, \distunitariesy-2);
    }

     \newcommand{\figureLayer}[2]{
        \figureW{#1}{#2}
        \figureLayerS{#1 + \widthunitary + \distunitariesx}{#2}
        \draw[] (#1 - 0.5 * \distunitariesx, #2 - 0.5 * \distunitariesx) rectangle ++(2*\widthunitary + 2 * \distunitariesx, \heightS + \distunitariesx);
        \draw[] (#1 - 0.5 * \distunitariesx, #2 + \heightS + \distunitariesy - 0.5 * \distunitariesx) rectangle ++(2*\widthunitary + 2 * \distunitariesx, \heightS + \distunitariesx);
    }

    \newcommand{\figureGreenUnder}[5]{
        \draw[color = colorgreen] (#1, #2) -- ++(#3, 0);
        \draw[color = colorgreen] (#1, #2) -- ++(0, #4);
        \draw[color = colorgreen] (#1 + #3, #2) -- ++(0, #4);
        \node[below, color = colorgreen] at (#1 + 0.5 * #3, #2) {#5};
    }

    \newcommand{\figureGreenLeft}[5]{
        \draw[color = colorgreen] (#1, #2) -- ++(0, #3);
        \draw[color = colorgreen] (#1, #2) -- ++(#4, 0);
        \draw[color = colorgreen] (#1, #2 + #3) -- ++(#4, 0);
        \node[left, color = colorgreen] at (#1 , #2 + 0.5 * #3) {#5};
    }
  

  \figureline{1}
  \figureline{4.5}
  \figureline{6.7}
  \figureline{1 + \distunitariesy + \heightS}
  \figureline{4.5+ \distunitariesy + \heightS}
  \figureline{6.7+ \distunitariesy + \heightS}
  
  \figureLayer{0}{0}
  \figureW{\xtwo}{0}
  \figureLayer{\xthree}{0}
  \figureW{\xfour}{0}

  \draw[dotted] (\xtwo + \widthunitary + \distunitariesx, \heightS + 0.5 * \distunitariesy) -- (\xthree - \distunitariesx,  \heightS + 0.5 * \distunitariesy);
  \draw[dotted] (-3, \heightS + 1) -- ++(0, \distunitariesy-2);

  \draw[dotted] (6, -2) -- ++(6, 0);

  \node[below] at (2.5, -0.75) {Layer 1};
  \node[below] at (15.5, -0.75) {Layer L};

  \node[above, color = colorred] at (0.5 * \widthunitary - 0.5, \heightS * 2 + \distunitariesy + 0.5) {$W_{\theta}^{(1)}$};
  \node[above, color = colorred] at (\xtwo + 0.5 * \widthunitary + 0.5, \heightS * 2 + \distunitariesy + 0.5) {$W_{\theta}^{(2)}$};
  \node[above, color = colorred] at (\xthree + 0.5 * \widthunitary - 0.5, \heightS * 2 + \distunitariesy + 0.5) {$W_{\theta}^{(L)}$};
  \node[above, color = colorred] at (\xfour + 0.5 * \widthunitary + 1.5, \heightS * 2 + \distunitariesy + 0.5) {$W_{\theta}^{(L+1)}$};

  \node[above, color = colorblue] at (1.5 * \widthunitary + \distunitariesx, \heightS * 2 + \distunitariesy + 0.5) {$S_1(\mathbf{x})$};
  \node[above, color = colorblue] at (\xthree + 1.5 * \widthunitary + \distunitariesx, \heightS * 2 + \distunitariesy + 0.5) {$S_L(\mathbf{x})$};
  
  \figureGreenUnder{-1}{-3}{20}{1}{$L$}
  \figureGreenLeft{-5}{-0.2}{8}{1}{$R$}
  \figureGreenLeft{-5}{\heightS + \distunitariesy-0.2}{8}{1}{$R$}
  \figureGreenLeft{-7}{-0.2}{20}{1}{$N \cdot R$}
\end{tikzpicture}
    }
    }
    ~ 
    \subfloat[\centering Sequential ansatz]{\resizebox{0.55\linewidth}{!}{
       \begin{tikzpicture}[scale = 0.25, line width = 0.5]
    \definecolor{colorblue}{HTML}{007AFF}
    \definecolor{colorred}{HTML}{D20000}
    \definecolor{colorgreen}{HTML}{7EC636}
   
    \def \roundness{3};
    \def \roundnesssubgenerators{1};
    \def \widthunitary{2};
    \def \widthsubgenerator{1.3};
    \def \heightS{8};
    \def \distunitariesx{1};
    \def \distunitariesy{4};

    \def \xtwo{2 *\widthunitary + 2 * \distunitariesx};
    \def \xthree{8};
    \def \xfour{\xthree + 2 *\widthunitary + 2 * \distunitariesx};

    \coordinate (origin) at (0, 0);
    \coordinate (S) at (\widthunitary, \heightS);
    \coordinate (subgenerator) at (\widthsubgenerator, \widthsubgenerator);
    
    \newcommand{\figureline}[1]{
    \draw (-1, #1) -- (33.5, #1);
    \node[left] at (-1, #1) {$|0\rangle$};
    }

    \newcommand{\figureW}[2]{
        \draw[rounded corners = \roundness, fill = colorred] (#1, #2) rectangle ++(S);
    }

    \newcommand{\figureS}[2]{
        \draw[rounded corners = \roundness, fill = colorblue] (#1, #2) rectangle ++(S);
        \draw[rounded corners = \roundnesssubgenerators, fill = white] (#1 + 0.5*\widthunitary-0.5*\widthsubgenerator, #2 + 0.5) rectangle ++(subgenerator);
        \draw[fill = white] (#1 + 0.5*\widthunitary-0.5*\widthsubgenerator, #2 + 6.2) rectangle ++(subgenerator);
        \draw[fill = white] (#1 + 0.5*\widthunitary-0.5*\widthsubgenerator, #2 + 4) rectangle ++(subgenerator);
        \draw[dotted] (#1 + 0.5 * \widthunitary, #2 + 2.5) -- ++(0, 1);
    }

     \newcommand{\figureLayer}[2]{
        \figureW{#1}{#2}
        \figureS{#1 + \widthunitary + \distunitariesx}{#2}
        \draw[] (#1 - 0.5 * \distunitariesx, #2 - 0.5 * \distunitariesx) rectangle ++(2*\widthunitary + 2 * \distunitariesx, \heightS + \distunitariesx);
    }
    \newcommand{\figureModel}[2]{
        \figureLayer{#1}{#2}
        \figureLayer{#1 + 8}{#2}
        \draw[dotted] (#1 + 2 * \widthunitary + \distunitariesx + 1, #2 + 0.35 * \heightS) -- ++(1.2, 0);
    }

    \newcommand{\figureGreenUnder}[5]{
        \draw[color = colorgreen] (#1, #2) -- ++(#3, 0);
        \draw[color = colorgreen] (#1, #2) -- ++(0, #4);
        \draw[color = colorgreen] (#1 + #3, #2) -- ++(0, #4);
        \node[below, color = colorgreen] at (#1 + 0.5 * #3, #2) {#5};
    }

    \newcommand{\figureGreenLeft}[5]{
        \draw[color = colorgreen] (#1, #2) -- ++(0, #3);
        \draw[color = colorgreen] (#1, #2) -- ++(#4, 0);
        \draw[color = colorgreen] (#1, #2 + #3) -- ++(#4, 0);
        \node[left, color = colorgreen] at (#1 , #2 + 0.5 * #3) {#5};
    }
  

  \figureline{1}
  \figureline{4.5}
  \figureline{6.7}
  
  \figureModel{0}{0}
  \figureModel{16}{0}
  \figureW{30.5}{0}

  \node[below] at (2.5, -0.75) {Layer 1};
  \node[below] at (10.5, -0.75) {Layer L};
  \node[below] at (2.5 + 16, -0.75) {Layer 1};
  \node[below] at (10.5 + 16, -0.75) {Layer L};

  \node[above, color = colorred] at (0.5 * \widthunitary - 0.5, \heightS + 1.5) {$W_{\theta}^{(1, 1)}$};
  \node[above, color = colorred] at (9, \heightS + 1.5) {$W_{\theta}^{(L, 1)}$};
  \node[above, color = colorred] at (0.5 * \widthunitary - 0.5+16, \heightS + 1.5) {$W_{\theta}^{(1, N)}$};
  \node[above, color = colorred] at (9+16, \heightS + 1.5) {$W_{\theta}^{(L, N)}$};
  \node[above, color = colorred] at (8+16+10, \heightS + 1.5) {$W_{\theta}^{(L+1, N)}$};

  \node[above, color = colorblue] at (1.5 * \widthunitary + \distunitariesx, \heightS + 0.5) {$S_{1,1}(x_1)$};
  \node[above, color = colorblue] at (12, \heightS + 0.5) {$S_{L, 1}(x_1)$};

  \node[above, color = colorblue] at (1.5 * \widthunitary + \distunitariesx + 16, \heightS + 0.5) {$S_{1, N}(x_N)$};
  \node[above, color = colorblue] at (12 + 16, \heightS + 0.5) {$S_{L, N}(x_N)$};

  \figureGreenLeft{-4}{-0.2}{8}{1}{$R$}
  \figureGreenUnder{-1}{-3}{15}{1}{$L$}
  \figureGreenUnder{-1 + 16}{-3}{15}{1}{$L$}
  \figureGreenUnder{-1}{-5}{31}{1}{$N \cdot L$}

   \draw[dotted] (14, 2.8) -- ++(1.2, 0);
   
\end{tikzpicture}
    }
    }
    \caption{Circuits of the parallel and the sequential ansatz. $R$ is the number of qubits per variable, $L$ the number of layers per variable $x_n$ and $N$ is the dimension of $\bx \in \R^N$. Note that in total the parallel ansatz needs $N \cdot R$ many qubits and $L$ layers in total, while the model for the sequential ansatz has $R$ qubits and $N\cdot L$ layers in total. The parameter encoding layers are coloured red, the data encoding layers are coloured blue. The data encoding layers have the form $S(x) = e^{-ixH}$ for some Hamiltonian $H$ called generators. The generators are typically composed of smaller sub-generators. We have illustrated that all generators are composed of $2\times 2$ matrices and thus acting on a single qubit each. The sub-generators have been marked as white squares.}
    \label{fig:qnn circuits parallel and sequential}
\end{figure}

For univariate functions, the two cases collapse into a single ansatz.

The QNN is trained iteratively by evaluating $f_{\btheta}(\bx)$ on a dataset $\mathcal{D}$ on quantum hardware, computing some loss of the output and updating the parameters $\btheta$ by some classical algorithm. To update the parameters, one can use the parameter shift rule \cite{Schuld2018EvaluatingAG} to compute the gradient with respect to $\btheta$ efficiently by two further evaluations of the QNN and use it for gradient descent.

An immediate question to ask is what kind of functions are represented by the class of QNN.
In \cite{Schuld2020EffectOD} it was shown that every QNN can be represented by a finite Fourier series
\begin{align}
    f_{\btheta}(\bx) = \sum_{\bomega \in \bOmega} c_{\bomega}(\btheta) e^{i \bomega \cdot \bx},
\end{align}
where $\bomega \cdot \bx \in \R$ denotes the standard scalar product of $\bomega$ and $\bx$, and $\bOmega \subseteq \R^N$ is a finite set called the \textbf{frequency spectrum}. It can be shown that the frequency spectrum only depends on the data encoding Hamiltonians $H_{n, l}$ and is the same for the parallel and sequential ansatz, while the Fourier coefficients $c_{\bomega}(\btheta)$ also depend on the parameter encoding circuits $W_{\btheta}^{(l)}$ and the observable $M$. 
More precisely, let $\Delta X := X-X := \lset{\lambda_1-\lambda_2|\ \lambda_1,\lambda_2 \in X}$ and $\sum_{l=1}^L X_l:=\lset{\sum_{l=1}^L \lambda_l|\ (\lambda_1,...,\lambda_L) \in X_1 \times \cdots \times X_L}$ for all sets $X, X_1,...,X_L\subseteq \R^N$, and let $\sigma\left(H_{l}\right) \subseteq \R$ be the spectrum of $H_l$, then the frequency spectrum of the univariate model is given by
\begin{align}
    \Omega  = \Delta \sum_{l=1}^L \sigma\left(H_l\right).
\end{align}
The frequency spectrum of the multivariate model $\bOmega$, regardless whether it is constructed by the parallel or the sequential approach, is given by the Cartesian product  $\bOmega = \Omega_1\times \cdots \times \Omega_N$ of the frequency spectra $\Omega_n = \Delta\sum_{l=1}^L \sigma\left(H_{n, l}\right) $ of the associated univariate models. Thus, maximising $\bOmega$ can be reduced to maximising all univariate spectra $\Omega_n$.

In this work, we focus on the frequency spectrum and present various maximality results for different architectures and assumptions made on the QNN.
There are two natural notions of maximality in this context. One is to maximise the size $|\Omega|$ of the frequency spectrum, the other is to maximize the number $K \in \N$ such that
\begin{align}
    \Z_K:= \lset{-K,...,0,...,K} \subseteq \Omega
\end{align}
to ensure proper approximation properties. In some cases, as we will see, the answers to these two questions are related. A useful concept to study maximality is the degeneracy of the quantum model \cite{Peters2022GeneralizationDO, Kordzanganeh2022AnEF}. The \textbf{degeneracy} $\deg(\omega) \in \N$ of the frequency $\omega \in \Omega$ of a given QNN is defined as the number of representations
\begin{align}
    \omega = \sum_{l=1}^L \left(\lambda_{k_l}^{(l)}-\lambda_{j_l}^{(l)}\right),
\end{align}
where $\lambda_1^{(l)},...,\lambda_d^{(l)} \in \R$ denotes the eigenvalues of $H_l$. Maximizing the set $\Omega$ is equivalent to minimizing the degeneracies.

Schuld et al. \cite{Schuld2020EffectOD} have shown further that for any square integrable function there exists a QNN approximating that function with the given precision.
More precisely, they have shown that for any square-integrable function $g\in L_2\left([0,2\pi]^N\right)$ and all $\epsilon>0$ there exists an observable $M_{g, \epsilon}$ and a single layer QNN
\begin{align}
    f(\bx) = \langle 0| \left(W^{(1)}\right)^\dagger S^\dagger(\bx) \left(W^{(2)}\right)^\dagger M_{g, \epsilon}W^{(2)} S(\bx) W^{(1)} |0\rangle
\end{align}
 with a sufficiently large number of qubits $R \in \N$ such that $\norm{f-g}_2 <\epsilon$ under some mild constraints on the generators of the data encoding layers called \textbf{universal Hamiltonian property}. This property guarantees that the frequency spectrum is rich enough to approximate the target function.  For example, if all generators are constructed out of Pauli matrices, this condition is satisfied. This underlines the importance of studying the properties of the frequency spectrum of QNNs. The broader the frequency spectrum, the better the potential for approximating target functions, at least if there is enough degree of freedom in the parameter circuits $W^{(l)}$. 
It should be noted, however, that although this universality theorem undoubtedly answers an important question regarding the class of functions that can be approximated by QNNs, the dependence of the observables M on the Fourier coefficients of the approximated function g is a too strong restriction to be an analogue of the universal approximation property of classical neural networks \cite{Hornik1989MultilayerFN}. In practice, M is typically set to $M=Z\otimes \Id \otimes \cdots \otimes \Id$ or similar with sufficiently many copies of the identity. In theory, it could be that a function cannot be well approximated with this class of observables, no matter how many qubits are used.
So as far as we know, it is still an open question whether this class of observables can approximate any square integrable function with arbitrary precision, or, more generally, whether there exists a sequence of observables $\left(M_k\right)_{k \in \N}$ such that for all target functions $g$ and all $\epsilon > 0$ there exists an index $K$ such that for all $k\geq K$ there exist unitaries $W^{(1)}, W^{(2)}$ such that $\norm{g-f}_2 < \epsilon$, where $M_k$ is used as observable in $f$ and $W^{(1)}, W^{(2)}$ as parameter unitaries.

\subsection{Related Work and Our Contribution} 
In order to present the results of related work and our contribution on the frequency spectrum $\Omega$, we need to introduce some additional terms and notations regarding various assumptions that can be made on the QNN ansatz.
In \cite{Schuld2020EffectOD} it is assumed that the data encoding circuits are all equal in a QNN, i.e., $S_1(x)=...=S_L(x)$. We refer to this case as a \textbf{QNN with equal data encoding layers}. If we do not make this addition, then we do not make any further assumptions about the layers. Note that in a single-layer model both cases collapse into one.
Further, for practical and theoretical considerations, it is often assumed that the generators, and hence the data encoding circuits, are built up of smaller sub-generators, typically $S_l(x)=\otimes_{r=1}^R e^{-ixH_{r, l}} = e^{-ix \oplus_{r=1}^R H_{r, l}}$, where $H_{r, l} \in \End{\B}$ is Hermitian. 
This should not be confused with the formula for data encoding in the parallel ansatz. Here, the sub-generators are arranged in parallel to form the generator. This can be done in either the univariate or multivariate setting. In the parallel ansatz, generators are arranged in parallel to encode multiple variables, with one generator used for each univariate variable.

The sub-generators $H_{l, r}$ are often set to some multiple of the Pauli matrices $X, Y, Z$, where
\begin{align}
    X = \begin{pmatrix}
0 & 1 \\
1 & 0 
\end{pmatrix},
Y = \begin{pmatrix}
0 & -i \\
i & 0 
\end{pmatrix},
Z = \begin{pmatrix}
1 & 0 \\
0 & -1 
\end{pmatrix}.
\end{align}
If the QNN is only built by generators of dimension $k \times k$, we refer to the QNN as \textbf{QNN with $k$-dimensional sub-generators}. We write \textbf{QNN with Pauli sub-generators} when only multiples of Pauli matrices are used.
If not clear from the context, we name the sub-generators explicitly and write \textbf{QNN with $k$-dimensional sub-generators $H_{l, r}$}.

In \cite{Schuld2020EffectOD, Heimann2022LearningCO}, it was shown that the frequency spectrum of a single layer QNN with Pauli sub-generators $\nicefrac{Z}{2}$ is given by
\begin{align}
    \Omega = \Z_R = \lset{-R,...,0,...,R}.
\end{align}
In \cite{Peters2022GeneralizationDO}, this setting is called Hamming encoding.
If instead a single qubit QNN with Pauli sub-generators $\nicefrac{Z}{2}$ and equal data encoding layers is used, the frequency spectrum is
\begin{align}
    \Omega = \Z_L = \lset{-L,...,0,...,L}.
\end{align}
The result is symmetrical and linear in $R$ and $L$. 
While these results consider the special case of $L=1$ and $R=1$, we extend this result to arbitrary $L$, $R$ and $2$-dimensional sub-generators $H$. In this case, the frequency spectrum is given by
\begin{align}
    \Omega = (\lambda-\mu)\cdot \Z_{RL} = \lset{(\lambda-\mu) \cdot k|\ k \in \Z_{RL}},
\end{align}
where $\lambda, \mu \in \R$ are the two eigenvalues of $H$.
The previous results can be derived directly from this.

In \cite{Kordzanganeh2022AnEF}, two ans\"atze for an exponential encoding scheme were presented, leading both to the same frequency spectrum. In their so called \textbf{sequential exponential} ansatz, which is in our terminology a single qubit QNN with Pauli sub-generators, where each $H_l$ is set to $H_l= \beta_l \cdot \nicefrac{Z}{2}$ with 
\begin{align}
    \beta_l := \begin{cases}
        2^{l-1}, & \text{if } l<L \\
        2^{L-1}+1, & \text{if } l=L,
    \end{cases}
\end{align}
the resulting frequency spectrum is given by $\Omega = \Z_{2^L}$ and therefore exponential in $L$. The second ansatz suggested called \textbf{parallel exponential}, which is in our terminology a single layer QNN with Pauli sub-generators $H_r=\beta_r \cdot \nicefrac{Z}{2}$ with $\beta_r$ as before, only the variable names are interchanged. This ansatz leads to the frequency spectrum $\Omega = \Z_{2^R}$, which is exponential in $R$. Almost the same encoding strategy was presented in \cite{Peters2022GeneralizationDO} for the single layer QNN with Pauli encoding layers $H_r= \beta_r \cdot \nicefrac{Z}{2}$. The only difference here is that no exception is made for the last term, i.e. $\beta_l = 2^{l-1}$ for all $l=1,...,L$. The authors have named the approach \textbf{binary encoding strategy}.
The frequency spectrum for this ansatz is given by frequency spectrum $\Omega = \Z_{2^R-1}$, therefore, only the two terms at the boundary are omitted compared to the parallel exponential ansatz.

The same approach was chosen in \cite{Shin2022ExponentialDE, Peters2022GeneralizationDO} for single layer QNNs with Pauli sub-generators of the form $H_r=\beta_r \cdot \nicefrac{Z}{2}$, but here the generators were multiplied by powers of $3$. More precisely, $\beta_r = 3^{r-1}$. In \cite{Peters2022GeneralizationDO}, this ansatz is called \textbf{ternary encoding strategy}. The frequency spectrum in this case is given by $\Omega = \Z_{\frac{3^R-1}{2}}$. Unlike the previous approaches, this ansatz is maximal in both senses, meaning that there is no $\Omega'$ such that $|\Omega'| > |\Omega|$ and that there is no $K>\frac{3^R-1}{2}$ such that $\Z_K \subseteq \Omega'$ for a QNN with that ansatz. This is due to the fact that the degeneracy is $1$ for all frequencies $\omega \in \Omega\backslash \lset{0}$, hence each frequency is given by a unique combination of differences of the eigenvalues \cite{Peters2022GeneralizationDO}.
We extend these results to QNNs with arbitrary number of layers $L \geq 1$ and arbitrary $2$-dimensional Hermitian sub-generators $H_{l, r}$, both for the equal layer approach and the non equal layer approach. If the data encoding layers are equal, the maximal frequency spectrum in both senses is given by
\begin{align}
    \Omega_{\max} = \Z_{\frac{(2L+1)^R - 1}{2}},
\end{align}
if no restrictions to the layers are made, the maximal frequency spectrum is given by
\begin{align}
    \Omega_{\max} = \Z_{\frac{3^{RL} - 1}{2}}.
\end{align}
It is no coincidence that the results for the Hamming ansatz, the sequential and parallel ansatz, as well as the ternary encoding strategy with non equal data encoding layers are symmetrical in $L$ and $R$ and depend only on $L \cdot R$. 

We show that there exists a bijection 
\begin{align}
    \mathcal{B}_b : \lset{\text{QNN}_k|\ \text{QNN}_k \text{ has shape } (R, L)}
    \longrightarrow
    \lset{\text{QNN}_k|\ \text{QNN}_k \text{ has shape } (R', L')}
\end{align}
between the set of all QNNs with $R$ qubits and $L$ layers with $k$-dimensional sub-generators and the set of all QNNs with $R'$ qubits and $L'$ layers with $k$-dimensional sub-generators such that the frequency spectrum is invariant under that transformation, as long as $R\cdot L = R' \cdot L'$ holds.
We denote that observation as \textbf{spectral invariance under area-preserving transformations} and $A:= R \cdot L $ as the \textbf{area} of the QNN. The symmetries in the results mentioned above can be directly derived from this.
Additionally, the results for QNNs with any number of layers can be derived straight from the single-layer results and the spectral invariance under area-preserving transformations, if no assumptions such as equal data-coding layers on the structure are made.

The previous results only considered QNNs with 2-dimensional generators, i.e. Hermitian operators acting on a single qubit. In \cite{Peters2022GeneralizationDO, Kordzanganeh2022AnEF}, Peters et al. presented how to extend these results to single layer QNNs with a $d$-dimensional generator, which is nothing other than allowing an arbitrary data-encoding $H$. In this case, the frequency spectrum $\Omega$ is maximal in size with $|\Omega| = 2 \binom{d}{2} + 1$ if and only if the eigenvalues of $H$ are a so called \textbf{Golomb ruler}. A Golomb ruler is a set of real numbers $\lambda_1<...<\lambda_k$ such that each difference $\lambda_i - \lambda_j$ except $0$ only occur once. Since there are maximally $2 \binom{k}{2} + 1$ possible pairs, $\lambda_1<...<\lambda_k$ are a Golomb ruler if and only if
\begin{align}
    \big|\Delta \lset{\lambda_1,...,\lambda_k}\big| = 2 \binom{k}{2} + 1.
\end{align}
Hence, the frequency spectrum $\Omega$ is maximal in size because it is non-degenerate. 
We extend this result to arbitrary $k$-dimensional sub-generators and provide an instruction how to build a QNN whose frequency spectrum is maximal in size with
\begin{align}
    |\Omega_{\max}| = (4^q - 2^q + 1)^{\nicefrac{RL}{q}},
\end{align}
where $q=\log_2(k)$.

To maximize the number $K \in \N$ such that $\Z_K \subseteq \Omega$, we introduce a novel approach similar to the \textbf{turnpike problem}. The turnpike problem goes like this: given a multiset $\Delta S$, find $S$. 
We relax prerequisites of the problem and name it consequently \textbf{relaxed turnpike problem}. The relaxed turnpike problem asks for an $S$ of size $d$
such that $K(S):=\max \lset{K \in \N_0|\ \Z_K \subseteq \Delta S}$ is maximal over all sets of size $d$. We propose an algorithm with complexity $\mathcal{O}\left(d^{2d}\right)$ to solve the relaxed turnpike problem. With this we show that for $k=d$ the number $K \in \N$ such that $\Z_K \subseteq \Omega$ is maximal if and only if the eigenvalues of the generator $H$ are a solution of the relaxed turnpike problem. We further give a construction to extend this to arbitrary $k$ and $L$, yielding
\begin{align}
    \Z_{\frac{(2K+1)^{\nicefrac{RL}{q}}-1}{2}} \subseteq \Omega,
\end{align}
where $K:= K(\sigma(H))$ for some $k$-dimensional sub-generator $H$ whose eigenvalues are a solution of the relaxed turnpike problem and $q:= \log_2(k)$. However, in this case $K'=\frac{(2K+1)^{\nicefrac{RL}{q}}-1}{2}$ is not necessarily maximal.
For $k\leq 4$ and therefore especially for Pauli sub-generators, the Golomb approach and the turnpike approach lead to the same results and in particular reproduce the results for $k=2$.

We collect all mentioned results on the frequency spectrum and our contributions in Table \ref{tab:result overview}.

\begin{table}[ht]
\centering
\resizebox{\textwidth}{!}{%
\begin{tabular}{l c c c l c c c c l}
\toprule
Encoding strategy &
  R &
  L &
  $H$&
  $\beta_{r, l}$ &
  Equal &
  $\Omega$ &
  $|\Omega|$ &
  Maximal &
  Source \\ \toprule
Hamming &
  $1$ &
  $\N$ &
  $\nicefrac{P}{2}$ &
  $1$ &
  \cmark &
  $\Z_{L}$ &
  $2 L +1$ &
  \xmark &
  \cite{Schuld2020EffectOD, Heimann2022LearningCO} \\ \midrule
Hamming &
  $\N$ &
  $1$ &
  $\nicefrac{P}{2}$ &
  $1$ &
  - &
  $\Z_{R}$ &
  $2R +1$ &
  \xmark &
  \cite{Schuld2020EffectOD, Heimann2022LearningCO, Peters2022GeneralizationDO} \\ \midrule
Sequential exponential &
  $1$ &
  $\N$ &
  $\nicefrac{P}{2}$ &
  $1,2, 2^2,...,2^{L-1}+1$ &
  \xmark &
  $\Z_{2^L}$ &
  $2^{L+1}-1$ &
  \xmark &
  \cite{Kordzanganeh2022AnEF} \\ \midrule
Parallel exponential &
  $\N$ &
  $1$ &
  $\nicefrac{P}{2}$ &
  $1,2, 2^2,...,2^{R-1}+1$ &
  - &
  $\Z_{2^R}$ &
  $2^{R+1}+1$ &
  \xmark &
  \cite{Kordzanganeh2022AnEF} \\ \midrule
Binary &
  $\N$ &
  $1$ &
  $\nicefrac{P}{2}$ &
  $2^{r-1}$ &
  - &
  $\Z_{2^R-1}$ &
  $2^{R+1}-1$ &
  \xmark &
  \cite{Peters2022GeneralizationDO} \\ \midrule
Ternary &
  $\N$ &
  $1$ &
  $\nicefrac{P}{2}$ &
  $3^{r-1}$ &
  - &
  $\Z_{\frac{3^R-1}{2}}$ &
  $3^R$ &
  $|\Omega|, K$ &
  \cite{Shin2022ExponentialDE, Peters2022GeneralizationDO} \\ \midrule
Golomb &
  $\N$ &
  $1$ &
  $d$ &
  $1$ &
  - &
  varies &
  $2 \binom{d}{2}+1$ &
  $|\Omega|$ &
  \cite{Kordzanganeh2022AnEF, Peters2022GeneralizationDO} \\ \midrule
\textbf{Hamming} &
  $\N$ &
  $\N$ &
  $\nicefrac{P}{2}$ &
  $1$ &
  \cmark &
  $\Z_{R L}$ &
  $2 RL +1$ &
  \xmark &
  Our \\ \midrule
\textbf{Equal Layers} &
  $\N$ &
  $\N$ &
  $2$ &
  $(2L+1)^{l-1 + L \cdot(r-1)}$ &
  \cmark &
  $\Z_{\frac{(2L+1)^R-1}{2}}$ &
  $(2L+1)^R$ &
  $|\Omega|, K$ &
  Our \\ \midrule
\textbf{Ternary} &
  $\N$ &
  $\N$ &
  $2$ &
  $3^{l-1 + L(r-1)}$ &
  \xmark &
  $\Z_{\frac{3^{RL}-1}{2}}$ &
  $3^{RL}$ &
  $|\Omega|, K$ &
  Our \\ \midrule
\textbf{Golomb} &
  $\N$ &
  $\N$ &
  $k$ &
  $(2\ell(\sigma(H))+1)^{l-1 + L(r-1)}$ &
  \xmark &
  varies &
  $\left(4^q-2^q+1\right)^{\nicefrac{RL}{q}}$ &
  $|\Omega|$ &
  Our \\ \midrule
\textbf{Turnpike} &
  $\N$ &
  $\N$ &
  $k$ &
  $(2 K+1)^{l-1 + L(r-1)}$ &
  \xmark &
  varies &
  $\geq (2K+1)^{\nicefrac{RL}{q}}$ &
  $K$ / \xmark &
  Our \\ \bottomrule
\end{tabular}%
}
\caption{
Summary of the frequency spectra results for QNNs. Our contributions are highlighted with bold letters. 
If $R$ or $L$ are arbitrary for a given encoding scheme, we denote that with $\N$. 
All encoding schemes use sub-generators of the form 
$H_{l, r} = \beta_{l, r} \cdot H$. If $H\in \lset{\nicefrac{X}{2}, \nicefrac{Y}{2}, \nicefrac{Z}{2}}$ is the half of a Pauli matrix, we abbreviate that by $H = \nicefrac{P}{2}$. Otherwise, the dimension of the arbitrary Hermitian matrix $H$ is stated in column $H$. If we write $k$, an arbitrary power of $2$ is allowed. The column named "Equal" indicates if equal data encoding layers are required. If the encoding scheme is maximal, we denote the type of maximality in the maximal column, either it is maximal in size $|\Omega|$ or maximal in $K$ such that $\Z_K \subseteq \Omega$. In the Golomb encoding scheme the eigenvalues of $H$ are a Golomb ruler and $q:=\log_2(k)$ with $q|R$. Similarly, in the turnpike encoding scheme, the eigenvalues of $H$ are a solution to the relaxed turnpike problem. Only when $L=1$ and $k=d$ this scheme guaranteed to be maximal.}
\label{tab:result overview}
\end{table}

Overall, this work advances the theoretical understanding of spectral properties of quantum neural networks by clarifying the underlying symmetries, maximality conditions, and algebraic structure of their frequency spectra. Our results rigorously characterize which frequency ranges are accessible for a given class of ansätze and which structural limitations are imposed by the encoding.

At the same time, our numerical experiments demonstrate that the knowledge of the maximal frequency spectrum and the given area alone is not sufficient to predict or guarantee good trainability. In particular, ansätze with identical maximal spectra can exhibit markedly different optimization behavior. Consequently, while the maximality results provide a principled way to analyze and compare the expressive potential of different encodings, they do not constitute a complete or predictive criterion for ansatz selection or performance. Rather, they should be understood as a structural baseline that helps interpret the capabilities and limitations of a given model, but must be complemented by additional considerations related to optimization and training.

\subsection{Structure of the Paper}
The paper is structured as follows.
Section~\ref{sec:notations, definitions and setup} serves to define the notation and to introduce some necessary concepts such as the Kronecker and Minkowski sum, as well as to describe some of their useful properties.
In Section~\ref{sec: QNNs general}, we discuss the representation of Quantum Neural Networks (QNNs) as finite Fourier series and explain how the frequency spectrum can be expressed in terms of the generator's spectrum. This holds true for various selections of $N$ and $L$, encompassing both parallel and sequential approaches.
In Section~\ref{sec: Spectral Invariance Under Area-Preserving Transformations}, we show that the frequency spectrum of a QNN solely relies on the area $A=RL$ and remains independent of the specific assignment of sub-generators to layers and qubits. This observation elucidates the symmetry observed between single-qubit and single-layer QNNs commonly reported in the literature.
In Section~\ref{sec: Maximal Frequency Spectrum for 2-Dimensional Sub-Generators} we prove maximality results for the frequency spectrum of QNNs with $2$-dimensional sub-generators, while Section~\ref{sec:k-dimensional generators} addresses the case of arbitrary-dimensional sub-generators. In Section \ref{sec:examples}, we give some numerical examples and provide some practical insights.
\section{Notations}
\label{sec:notations, definitions and setup}
Let $\B$ denote the two-dimensional Hilbert space $\C^2 = \C\ket{0}+\C \ket{1}$ endowed with the standard inner product $\langle \cdot,\cdot\rangle$. For $R \in \N$ let $\BR$ be the $R$-fold tensor product of $\B$ and $\U_R:=\lset{ U \in \End{\BR} | UU^{\dagger}=\Id}$ the unitary group of $\BR$. With $R$ we always denote the number of considered qubits.
The overall dimension of $\BR$ is denoted by $d=2^R$.
Sometimes it will also be useful to consider the standard inner product on $\R^N$. To avoid confusion with the Dirac notation, we write 
$\bx \cdot \by = \bx_1\by_1 + \ldots + \bx_N\by_N$ for $\bx,\by \in \R^N$.

For $n \in \N_0 = \lset{0, 1, 2, 3,\ldots}$ we define
\begin{align*}
    [n] := \lset{0,1,\ldots,n-1} \subseteq \N_0
\end{align*} 
and
\begin{align*}
    \Z_n :=\lset{-n,\ldots,0,\ldots,n} \subseteq \Z.
\end{align*}
We make use of the natural identification of elements $j \in [d]$ with the vectors of the computational basis $|j\rangle$ of $\BR$.

The data encoding layers $S_l(x)$ of QNNs typically consist of smaller building blocks $S_l(x) = \bigotimes_{r=1}^R e^{-i x H_{r, l}}$ with sub-generators $H_{r, l}$. To see that $S_l(x)$ can be rewritten in the form $S_l(x) = e^{-ix\tilde{H}}$ for some generator $\tilde{H}$ as required, the following construction is useful.
\begin{definition}[Kronecker Sum]
   Let $V_r$ be some finite dimensional vector spaces and $H_r \in \End{V_r}$ be linear maps for all $r=1,...,R$. Define
   \begin{align*}
       H_r':= \Id_1\otimes\cdots \otimes \Id_{r-1}\otimes H_r \otimes \Id_{r+1} \cdots \otimes \Id_R,
   \end{align*}
   where $\Id_r \in \End{V_r}$ denotes the identity map.
   The \textbf{Kronecker sum} of the linear maps is defined as
   \begin{align*}
       \bigoplus_{r=1}^R H_r := \sum_{r=1} H_r'.
   \end{align*}
\end{definition}
The Kronecker sum has the following well-known properties:
\begin{lemma}[Properties of the Kronecker Sum]
\label{lemma:basic properties kronecker sum}
    Let $V_r$ be some finite dimensional vector spaces and $H_r \in \End{V_r}$ be Hermitian for all $r=1,...,R$.
    \begin{itemize}
    \item[(1)] It holds
    \begin{align*}
        \bigotimes_{r=1}^R e^{H_r} = e^{\oplus_{r=1}^R H_r}.
    \end{align*}
    In particular, $\bigotimes_{r=1}^R e^{-i x H_r} = e^{-i x \tilde{H}}$ for some Hermitian $\tilde{H}$ and all $x \in \R$.
    \item[(2)] 
    We have
    \begin{align*}
        \sigma\left(\bigoplus_{r=1}^R H_r\right) = \sum_{r=1}^R \sigma(H_r),
    \end{align*}
    where $\sigma(\cdot)$ denotes the spectrum of the operator.
\end{itemize}
\end{lemma}
As mentioned in the introduction, we are interested in the frequency spectra of QNNs, which can be represented via sums and differences of the sets of eigenvalues of the sub-generators. For this reason, the following notations are useful in this work.
\begin{definition}
    Let $A, A_1, A_2, \ldots, A_n \subseteq \R^N$ be arbitrary subsets.
    \begin{itemize}
        \item[(a)] The \textbf{Minkowski sum} of $A_1, \ldots, A_n$ is defined as 
        \begin{align*}
            \sum_{i=1}^n A_i := \lset{a_1 + \ldots + a_n | a_1 \in A_1, \ldots, a_n \in A_n}.
        \end{align*}
        \item[(b)] For any $r \in \R$, we define $r \cdot A := \lset{r \cdot a|\ a\in A}$.
        \item[(c)] The special case where $n=2$ and $A_2 = (-1)\cdot A_1$ is abbreviated
            as $\Delta A := \lset{ a - b | a,b \in A}$.
        We also write $A_1 \Delta A_2 := (A_1-A_2) \cup (A_2-A_1)$, which is not to be confused with the symmetric difference of sets, which we do not use in this work.
        
    \end{itemize}
\end{definition}
We note some well-known properties of the Minkowski sum.
\begin{lemma}
\label{lemma:minkowski properties}
    Let $A_1, A_2, \ldots, A_n, B_1, B_2, \ldots,B_n \subseteq \R^N$ be arbitrary subsets.
    \begin{itemize}
        \item[(a)] $\sum_{i=1}^n \Delta A_i = \Delta(\sum_{i=1}^n A_i)$.
        \item[(b)] $\Delta(A_1 \times \ldots \times A_n) = \Delta A_1 \times \ldots \times \Delta A_n$.
        \item[(c)] $\sum_{i=1}^n (A_i \times B_i) = (\sum_{i=1}^n A_i) \times (\sum_{i=1}^n B_i)$.
    \end{itemize}
\end{lemma}
\section{Quantum Neural Networks and Fourier Series}
\label{sec: QNNs general}
Let us first recall the definition of Quantum Neural Networks. For that, we consider two types of ans\"atze.
\begin{definition}[Parallel and Sequential Ansatz]
\label{def:ansaetze}
Let $R,L\in \N$ and $\bx \in \R^N, \btheta \in \R^P$.
    \begin{itemize}
        \item[(a)] A \textbf{parallel ansatz} is a parametrized circuit of the form
        \begin{align*}
                U(\bx,\btheta) = W^{(L+1)}_{\btheta}S_L(\bx)W^{(L)}_{\btheta}\cdots W^{(2)}_{\btheta} 
                S_1(\bx)W^{(1)}_{\btheta}
        \end{align*}
        with
        \begin{align*}
            S_l(\bx) = \bigotimes_{n=1}^N e^{-i x_n H_{l, n}},
        \end{align*}
        where each $H_{l, n} \in \End{\BR}$ is Hermitian and each $W^{(l)}_{\btheta} \in \End{\B^{\otimes R\cdot N}}$ is unitary.

        \item[(b)]
        A \textbf{sequential ansatz} is a parametrized circuit of the form
        \begin{align*}
            U(\bx, \btheta) = \prod_{n=1}^N U_n(\bx_n, \btheta),
        \end{align*}
        with
        \begin{align*}
            U_n(\bx_n, \btheta) = 
            W^{(L+1, n)}_{\btheta}S_L(x_n)W^{(L, n)}_{\btheta}\cdots W^{(2, n)}_{\btheta} 
            S_1(x_n)W^{(1, n)}_{\btheta}
        \end{align*}
        and
        \begin{align*}
           S_l(x_n) = e^{-ix_n H_{l, n}}.
        \end{align*}
        Again, each $H_{l, n} \in \End{\BR}$ is Hermitian and each $W^{(l, n)}_{\btheta} \in \End{\BR}$ is unitary.
    \end{itemize}
\end{definition}
In both cases, the Hermitians $H_{l, n}$ are called the \textbf{generators} of the ansatz.
We briefly mention some simplifying assumptions we can make when working with these ans\"atze.
First, we can omit the dependence of $W^{(l)}_{\btheta}$ on $\btheta$ from the notation since we are not applying any optimization process and assume that all unitaries $W \in \U_R$ can be represtend by some $W_{\btheta}$.
Therefore we write $W^{(l)}:=W^{(l)}_{\btheta}$ from now on.

Second, since the $H_{l, n}$ are Hermitian, there exists unitaries $U_{l, n}\in \U_R$ such that $U_{l, n}^\dagger H_{l, n} U_{l, n} = \diag\left(\lambda_0^{(l, n)}, ..., \lambda_{d-1}^{(l, n)}\right)$ in the computational basis $|0\rangle, |1\rangle,...,|d-1\rangle \in \BR$. 
In the sequential ansatz, we can "absorb" the $U_{l, n}$ in the unitaries $W^{(l, n)}$ via $\tilde{W}^{(l, n)} := U_{l, n} W^{(l, n)} U_{l-1, n}^\dagger$. 
In the parallel ansatz, we can use the unitary 
$\bigotimes_{n=1}^N U_{l, n}$
and let it be absorbed by the $W^{(l)}$.
Hence, without loss of generality, the $H_{l, n}$ are diagonal with eigenvalues 
\begin{align*}
    \lambda_0^{(l, n)}, ..., \lambda_{d-1}^{(l, n)} \in \R.
\end{align*}

\begin{definition}[Quantum Neural Networks]
\label{def:qnn}
    Let $U(\bx)$ be a parallel or sequential ansatz. A \textbf{Quantum Neural Network (QNN)} is a function of the form
    \begin{align*}
        f: \R^N &\longrightarrow \R \\
        \bx &\longmapsto 
        \langle 0| U^\dagger(\bx) M U(\bx)|0\rangle,
    \end{align*}
    with some Hermitian $M$. We call the tuple $(R, L) \in \N^2$ the \textbf{shape} of the QNN. 
\end{definition}

In \cite{Schuld2020EffectOD} it was shown that QNNs with equal data encoding layers can be written in the form
\begin{align*}\label{eq:fourier_series}
    f(\bx) = \sum_{\omega \in \bOmega} c_{\bomega}e^{i \bx \cdot \bomega},
\end{align*}
where $c_{\bomega} \in \C$\ and $\bOmega \subseteq \R^N$ is a finite set, called the \textbf{frequency spectrum}. If $\bOmega \subseteq \Z^N$, then this sum is the partial sum of a Fourier series.
We reformulate the above results for univariate models in the following theorem.
Recall that for $N=1$, the parallel and sequential ansatz are equal and we write $H_l:=H_{l, 1}$.

\begin{theorem}[Univariate QNN is a Fourier Series]
\label{theorem: Univariate QNN is Fourier}
Let $f(x) = \langle 0 | U^{\dagger}(x) M U(x)|0\rangle$ be a univariate QNN with arbitrary generators $H_{l} \in \End{\B^{\otimes R}}$.
Then
\begin{align*}
    f(x) = \sum_{\omega  \in \Omega} c_{\omega} e^{-i \omega \cdot x},
\end{align*}
    with
\begin{align*}
    \Omega 
    = \sum_{l=1}^L \Delta \sigma (H_l).
\end{align*}
\end{theorem}
The result above can be generalized to multivariate QNNs. We skip the technical details and just give the final result here. The details and a proof of the previous theorem and the following result can be found in Appendix \ref{appendix:frequency spectrum of multivariate models}.
\begin{theorem}[Frequency Spectrum of  a Multivariate QNN]
\label{theorem: frequency spectrum of multivariate QNN is rectangular}
The frequency spectrum $\bOmega_{L, N}$ of a multivariate QNN with generators $H_{l, n} \in \End{\BR}$, regardless of whether the parallel or sequential approach is chosen, is given by
\begin{align*}
        \bOmega_{L, N} = 
        \Omega_{L}\left(H_{1, 1},...,H_{L, 1}\right) \times \cdots \times \Omega_{L}\left(H_{1, N},...,H_{L, N}\right),
\end{align*}
where $\Omega_{L}\left(H_{1, n},...,H_{L, n}\right)$ denotes the frequency spectrum of the univariate model with generators $H_{1, n},...,H_{L, n}$ for all $n=1,\ldots,N$.
\end{theorem}
According to Theorem \ref{theorem: frequency spectrum of multivariate QNN is rectangular}, the frequency spectrum of a multivariate QNN depends only on the generators of the ansatz. Specifically, it is rectangular, that is, the Cartesian product of the corresponding univariate frequency spectra $\Omega_{L}\left(H_{n, 1},...,H_{n, L}\right)$ to the variable $x_n$. In the following we are only interested in the frequency spectrum of a QNN. In particular, our goal is to maximise it, so by Theorem \ref{theorem: frequency spectrum of multivariate QNN is rectangular}, maximising the multivariate frequency spectrum is equivalent to maximising each univariate frequency spectrum. It is therefore sufficient to consider only univariate models, i.e. $N=1$, in the following.
\section{Spectral Invariance Under Area-Preserving Transformations}
\label{sec: Spectral Invariance Under Area-Preserving Transformations}
We have seen that the frequency spectrum of a QNN depends only on the eigenvalues of the generators $H_{l}$, or, if it is composed of smaller sub-generators $H_{r, l}$, on their eigenvalues. In this section, we show how to transform the QNN in such a way that the frequency spectrum remains invariant and the resulting QNN uses the same sub-generators. More precisely, given a QNN of shape $(R, L)$, we show that one can rearrange the sub-generators without changing the frequency spectrum as long as the \textbf{area} $A:= R \cdot L$ of the QNN is preserved.
We summarise the previous idea in the following theorem.
\begin{theorem}[Spectral Invariance Under Area-Preserving Transformations]
\label{theorem: Spectral Invariance Under Area-Preserving Transformations}
Let $k=2^q$ and
let $R, R', L, L' \in \N$ with $R \cdot L = R' \cdot L'$ such that $q$ divides $R$ and $R'$. Then every bijection 
\begin{align}
    b: \left[\nicefrac{R'}{q}\right] \times [L'] \rightarrow \left[\nicefrac{R}{q}\right] \times [L]
\end{align}
induces a bijection
\begin{align}
    \mathcal{B}_b : \lset{\text{QNN}_k|\ \text{QNN}_k \text{ has shape } (R, L)}
    \longrightarrow
    \lset{\text{QNN}_k|\ \text{QNN}_k \text{ has shape } (R', L')}
\end{align}
between univariate QNNs which only consists of $k$-dimensional sub-generators such that the frequency spectrum of the models is invariant under that transformation. 

In particular, if we make no further assumptions on the generators like equal data encoding layers, the maximal possible frequency spectrum of a QNN with $k$-dimensional sub-generators is only dependent on the \textbf{area} $A(R, L):=R \cdot L \in \N$ and not on the individual $R, L$.
We call $\mathcal{B}_b$ an \textbf{area-preserving transformation}.
\end{theorem}
\begin{proof}
    Let $k=2^q$ and $H_{r, l} \in \End{\B^{\otimes q}}$ be the $k$-dimensional sub-generators of a univariate QNN of shape $(R, L)$.
    By Theorem \ref{theorem: Univariate QNN is Fourier} and Lemma \ref{lemma:minkowski properties}, the frequency spectrum of that QNN is given by 
    \begin{align*}
     \Omega = \sum_{l=1}^L \sum_{r=1}^{\nicefrac{R}{q}} \Delta \sigma(H_{r, l}).
 \end{align*}
 The QNN of shape $(R', L')$ with sub-generators $H_{b(r-1, l-1) + (1, 1)}$, then has trivially the same frequency spectrum, as it is only a permutation of the sub-generators and the Minkowski sum remains the same. We only needed to perform an index shift, since $r$ ranges from $1, \ldots, R/q$ and $l$ from $1, \ldots, L$, whereas $[R/q] = \{0, \ldots, R/q - 1\}$ and $[L] = \{0, \ldots, L - 1\}$.

\end{proof}
Note that two QNNs are considered equal if and only if they have the same shape and all their sub-generators $H_{r, l}$ and $H_{r, l}'$ for all $r=1,\ldots , R$ and $l=1,\ldots , L$ are equal. Strictly speaking, however, the QNNs also depend on the parameter encoding layers. Therefore, we have only constructed a bijection between equivalence classes of QNNs, which is sufficient for our cases here. Using the Schr\"oder-Bernstein theorem of set theory and the axiom of choice, it is possible to extend this to a bijection of QNNs considering the parameter encoding layers.

By Theorem \ref{theorem: Univariate QNN is Fourier}, the only relevant information for the frequency spectrum are the generators or sub-generators. Thus, the QNN can be represented by a rectangular plot containing these sub-generators as squares in the arrangement matching their occurrences in the circuit. Theorem \ref{theorem: Spectral Invariance Under Area-Preserving Transformations} now implies that this arrangement is irrelevant, since we can transform each QNN with an area-preserving transformation in all other compatible rectangles. We sketch this in Figure \ref{fig:invariance under area-preserving transformation}.
\begin{figure}
    \centering
    \resizebox{0.95\linewidth}{!}{
        \begin{tikzpicture}[scale = 1, line width = 1.5]
    \definecolor{colorblue}{HTML}{007AFF}
    \definecolor{colorred}{HTML}{D20000}
    \definecolor{colorgreen}{HTML}{7EC636}
   
    \def \roundness{3};

    \coordinate (origin) at (0, 0);

    \newcommand{\figureGreenUnder}[5]{
        \draw[color = colorgreen] (#1, #2) -- ++(#3, 0);
        \draw[color = colorgreen] (#1, #2) -- ++(0, #4);
        \draw[color = colorgreen] (#1 + #3, #2) -- ++(0, #4);
        \node[below, color = colorgreen] at (#1 + 0.5 * #3, #2) {#5};
    }

    \newcommand{\figureGreenLeft}[5]{
        \draw[color = colorgreen] (#1, #2) -- ++(0, #3);
        \draw[color = colorgreen] (#1, #2) -- ++(#4, 0);
        \draw[color = colorgreen] (#1, #2 + #3) -- ++(#4, 0);
        \node[left, color = colorgreen] at (#1 , #2 + 0.5 * #3) {#5};
    }
    
    \foreach \x in {1, 2, 3, 4} {
        \foreach \y in {1, 2, 3} {
            \draw[fill = colorblue, fill opacity = 0.2, colorblue] (\x, -\y) rectangle node[black, fill opacity = 1.0] {$H_{\y, \x}$} ++(1, 1);
        }
    }
    \node[above] at (6.5, -1.5) {$\mathcal{B}_b$};
    \draw[<->, colorblue] (5.5, -1.5) -- ++(2,0);
    \foreach \x in {1, 2, 3, 4, 5, 6} {
        \foreach \y in {1, 2} {
            \draw[fill = colorblue, fill opacity = 0.2, colorblue] (\x+ 8.5, -\y-0.5) rectangle  ++(1, 1);
        }
    }
    \foreach \x in {1, 2, 3, 4} {
        \node[] at (9 + \x, -1) {$H_{1, \x}$};
    }
    \foreach \x in {1, 2} {
        \node[] at (9 + 4 + \x, -1) {$H_{2, \x}$};
    }
    \foreach \x in {3, 4} {
        \node[] at (9 - 2 + \x, -2) {$H_{2, \x}$};
    }
    \foreach \x in {1, 2, 3, 4} {
        \node[] at (9 + 2 + \x, -2) {$H_{3, \x}$};
    }

    \figureGreenLeft{0.5}{-3}{3}{0.3}{$\nicefrac{R}{q} = 3$}
    \figureGreenLeft{9}{-2.5}{2}{0.3}{$\nicefrac{R'}{q} = 2$}

    \figureGreenUnder{1}{-3.5}{4}{0.3}{$L = 4$}
    \figureGreenUnder{9.5}{-3}{6}{0.3}{$L' = 6$}
    
\end{tikzpicture}
    }
    \caption{Visualisation of Theorem \ref{theorem: Spectral Invariance Under Area-Preserving Transformations}. By Theorem \ref{theorem: Univariate QNN is Fourier}, the frequency spectrum only depends on the sub-generators $H_{r, l}$. The QNN can thus be represented by a rectangle containing these sub-generators as squares in the arrangement matching their occurrence in the quantum circuit. By Theorem \ref{theorem: Spectral Invariance Under Area-Preserving Transformations}, the arrangement is irrelevant for the frequency spectrum, as there exists an area-preserving transformation for all other compatible rectangles.}
    \label{fig:invariance under area-preserving transformation}
\end{figure}

If no requirements are made on the generators, such as equal data coding layers, Theorem \ref{theorem: Spectral Invariance Under Area-Preserving Transformations} implies that, without loss of generality, we can consider single-layer models to study the maximum frequency spectrum of QNNs with $k$-dimensional generators. The results depending on $R$ can then be generalised to the multi-layer case by simply replacing the dependence on $R$ by $R \cdot L$.

The technical requirement that $q$ divides $R$ is only necessary because we have made the assumption that the QNN consists only of sub-generators of dimension $k$, which is only possible if this is satisfied. We briefly sketch how to generalize Theorem \ref{theorem: Spectral Invariance Under Area-Preserving Transformations} allowing mixed-dimensional generators. In this case, the squares representing the sub-generators like in Figure \ref{fig:invariance under area-preserving transformation} are replaced by rectangles with side-length $q_{r, l} \times 1$, where $k_{r, l}=2^{q_{r, l}}$ is the dimension of the sub-generator $H_{r, l}$.
An area-preserving transformation thus is only possible into rectangular shapes which can be constructed out of these rectangular tiles, see Figure \ref{fig:invariance under area-preserving transformation, not k-dimensional} for an example.
\begin{figure}
    \centering
    \resizebox{0.75\linewidth}{!}{
        \begin{tikzpicture}[scale = 1, line width = 1.5]
    \definecolor{colorblue}{HTML}{007AFF}
    \definecolor{colorred}{HTML}{D20000}
    \definecolor{colorgreen}{HTML}{7EC636}

    \newcommand{\tile}[3]{
        \draw[fill = colorblue, fill opacity = 0.2, colorblue] (#1, -#2) rectangle ++(1, -#3);
    }
    \foreach \pos/\h in {0/3, 3/3} {
        \tile{0}{\pos}{\h}
    }
     \foreach \pos/\h in {0/2, 2/1, 3/1, 4/2} {
        \tile{1}{\pos}{\h}
    }
     \foreach \pos/\h in {0/1, 1/1, 2/4} {
        \tile{2}{\pos}{\h}
    }
     \foreach \pos/\h in {0/4, 4/2} {
        \tile{3}{\pos}{\h}
    }
     \foreach \pos/\h in {0/2, 2/1, 3/1, 4/2} {
        \tile{4}{\pos}{\h}
    }
    \foreach \pos/\h in {0/4, 4/2} {
        \tile{5}{\pos}{\h}
    }
    \node[above] at (7.5, -3) {$\mathcal{B}_b$};
    \draw[<->, colorblue] (6.5, -3) -- ++(2,0);

        \tile{9}{1}{4}
     \foreach \pos/\h in {0/2, 2/2} {
        \tile{9 + 1}{\pos - 1}{\h}
    }
    \foreach \pos/\h in {0/3, 3/1} {
        \tile{9 + 2}{\pos - 1}{\h}
    }
     \foreach \pos/\h in {0/2, 2/2} {
        \tile{9 + 3}{\pos - 1}{\h}
    }
     \foreach \pos/\h in {0/1, 1/1, 2/2} {
        \tile{9 + 4}{\pos - 1}{\h}
    }
        \tile{9 + 5}{1}{4}
        
     \foreach \pos/\h in {0/2, 2/2} {
        \tile{9 + 6}{\pos - 1}{\h}
    }
     \foreach \pos/\h in {0/2, 2/1, 3/1} {
        \tile{9 + 7}{\pos - 1}{\h}
    }
        \foreach \pos/\h in {0/1, 1/3} {
        \tile{9 + 8}{\pos - 1}{\h}
    }

\end{tikzpicture}
    }
    \caption{Example of how to extend Theorem \ref{theorem: Spectral Invariance Under Area-Preserving Transformations} to QNNs without the requirement that all generators are $k$-dimensional. The sides of the individual rectangles representing a $k_{r, l}=2^{q_{r, l}}$-dimensional sub-generator $H_{r, l}$ have side lengths $q_{r, l} \times 1$. In this example, we used $q=1, 2, 3, 4$, $(R, L) = (6, 6)$ and $(R', L') = (4, 9)$. An area-preserving transformation is only possible if the target rectangular could be tiled with the given sub-generator rectangles.}
    \label{fig:invariance under area-preserving transformation, not k-dimensional}
\end{figure}

We emphasise that the invariance is only for the frequency spectrum, we do not necessarily obtain the same Fourier series, since the parameter encoding layers $W^{(l)}$ and the observable $M$ remain unchanged and therefore have incompatible dimensions for other shapes $(R', L')$. Therefore, the parameter encoding layers and the observable would also need to be transformed for a bijection that preserves the full Fourier series. 
The consequence of a Fourier series preserving transformation would be that there would be no advantage in terms of expressibility and approximation in considering models with more than one layer or with more than one qubit. Thus, it is important to keep in mind that area-preserving transformations preserve only the frequency spectrum and not necessarily the entire Fourier series.

While the frequency spectrum depends only on the area $A$ and not on the concrete shape $(R, L)$, for practical applications the number of available qubits and error rates could play an important role in choosing a well-suited shape of the QNN.

In \cite{Kordzanganeh2022AnEF}, two exponential coding schemes were presented, namely the parallel exponential and the sequential exponential ansatz, both of which lead to the same frequency spectrum. This relationship can be explained using Theorem \ref{theorem: Spectral Invariance Under Area-Preserving Transformations}, as one ansatz is just the image under an area-preserving transformation of the other.
\section{2-Dimensional Sub-Generators}
\label{sec: Maximal Frequency Spectrum for 2-Dimensional Sub-Generators}
In this section we consider univariate QNNs with $2$-dimensional generators.
In \cite{Schuld2020EffectOD} it was shown that the frequency spectrum of a single layer model with Pauli sub-generators $\nicefrac{Z}{2}$ is given by $\Omega = \Z_R$, while the frequency spectrum of a single qubit QNN with the same sub-generators is given by $\Omega = \Z_L$. The symmetry in these results is a direct consequence of Theorem \ref{theorem: Spectral Invariance Under Area-Preserving Transformations}. If one allows an arbitrary $2$-dimensional sub-generator $H$ to replace $\nicefrac{Z}{2}$ and arbitrary many layers and qubits, the resulting QNN has the following frequency spectrum.
\begin{theorem}[Frequency Spectrum of the Hamming Encoding Strategy]
    Let $H \in \End{\B}$ be Hermitian and $ \lambda, \mu \in \R$ its eigenvalues.
    The frequency spectrum $\Omega$ of the univariate QNN with $2$-dimensional sub-generators $H_{r, l} = H$ is given by
    \begin{align}
        \Omega = (\lambda-\mu) \cdot \Z_{R L}.
    \end{align}
\end{theorem}
\begin{proof}
    By Theorem \ref{theorem: Spectral Invariance Under Area-Preserving Transformations}, w.l.o.g. we can assume $R=1$ and replace $L$ by $R\cdot L$ afterwards. By Theorem \ref{theorem: Univariate QNN is Fourier} and Lemma \ref{lemma:basic properties kronecker sum}, the frequency spectrum is given by
    \begin{align}
        \Omega = \sum_{l=1}^L \Delta \sigma(H) = \sum_{l=1}^L (\lambda-\mu) \cdot \Z_1 = (\lambda-\mu)\cdot \Z_{L}.
    \end{align}
\end{proof}
In \cite{Peters2022GeneralizationDO, Shin2022ExponentialDE} it was shown that the maximal frequency spectrum in both senses of a single layer QNN with Pauli sub-generators $\nicefrac{P}{2}$ is given by $\Omega_{\max} = \Z_{\frac{3^R-1}{2}}$, i.e. it is maximal in size and maximal in $K$ such that $\Z_K \subseteq \Omega$. For the ansatz with no further restrictions to the data encoding layers, we can extend this result directly by Theorem \ref{theorem: Spectral Invariance Under Area-Preserving Transformations} to QNNs with arbitrary many layers $L\geq 1$, i.e. $\Omega_{\max} = \Z_{\frac{3^{R\cdot L}-1}{2}}$. However, we follow another approach. We first show that the maximal frequency spectrum of a univariate QNN with equal data encoding layers and arbitrary $2$-dimensional generators is given by $\Omega_{\max} = \Z_{\frac{(2L+1)^R-1}{2}}$. We then use Theorem \ref{theorem: Spectral Invariance Under Area-Preserving Transformations} to extend this result to QNNs without the restriction of equal data encoding layers.

The maximum frequency spectrum of QNNs with equal data encoding layers is given by the following theorem, a proof can be found in Appendix \ref{appendix: maximality proofs 2dim}.
\begin{theorem}[Maximal Frequency Spectrum with Equal Data Encoding Layers]
   \label{theorem: maximal frequency spektrum equal data encoding layers}
Let $\Omega$ be the frequency spectrum of a univariate QNN with equal data encoding layers and $2$-dimensional sub-generators $H_{r, l} := H_r \in \End{\B}$ for all $r = 1,\ldots,R$.
Further, let $\lambda_r, \mu_r \in \R$ be the two eigenvalues of $H_r$ and w.l.o.g. assume $0 \leq \lambda_1 - \mu_1 \leq ... \leq \lambda_R - \mu_R$.
 \begin{itemize}
    \item[(a)] Let $K \in \N$ be maximal with respect to the property
        \begin{align*}
            \Z_K \subseteq \Omega.
        \end{align*} 
    Then $K \leq \frac{(2L+1)^R -1}{2}$ with equality if and only if $\lambda_r - \mu_r = (2L+1)^{r-1}$ for all $r=1,...,R$. 

    \item[(b)] In particular, if one fixes $H \in \End{\B}$ with eigenvalues $\lambda, \mu \in \R$ and $\lambda - \mu = 1$, one can set $H_r := (2L+1)^{r-1} H$ to obtain
    \begin{align*}
        \Omega = \Z_{\frac{(2L+1)^R-1}{2}}.
    \end{align*}
\end{itemize}
\end{theorem}
We can easily extend the above results to QNNs without requiring all layers to be equal, by reducing to Theorem \ref{theorem: maximal frequency spektrum equal data encoding layers}, using Theorem \ref{theorem: Spectral Invariance Under Area-Preserving Transformations} and applying a simple trick. Given a model with arbitrary shape $(R, L)$, the frequency spectrum is the same as that of a model with shape $(R', L') = (R \cdot L, 1)$, as obtained by Theorem \ref{theorem: Spectral Invariance Under Area-Preserving Transformations}. Since this transformed model has only a single layer, all layers are trivially equal, and Theorem \ref{theorem: maximal frequency spektrum equal data encoding layers} is therefore applicable. The detailed proof can be found in Appendix \ref{appendix: maximality proofs 2dim}.

\begin{theorem}[Maximal Frequency Spectrum for Arbitrary Data Encoding Layers]
   \label{theorem: maximal frequency spektrum, arbitrary S}
 Let $\Omega$ be the frequency spectrum of a univariate QNN with $2$-dimensional sub-generators $H_{r, l} \in \End{\B}$ for all $r = 1,\ldots, R$.
Further, let $\lambda_r, \mu_r \in \R$ be the two eigenvalues of $H_r$ and w.l.o.g. assume $0 \leq \lambda_1 - \mu_1 \leq ... \leq \lambda_R - \mu_R$.
 \begin{itemize}
    \item[(a)] Let $K \in \N$ be maximal with respect to the property
    \begin{align*}
     \Z_K \subseteq \Omega.
 \end{align*} 
 Then $K \leq \frac{3^{RL} -1}{2}$ with equality if and only if 
 $\lset{\lambda_r^{(l)} - \mu_r^{(l)}|\ r, l} = \lset{3^0,3^1,\ldots, 3^{R\cdot L-1}}$. 
 
 \item[(b)] In particular, if one fixes $H \in \End{\B}$ with eigenvalues $\lambda, \mu \in \R$ and $\lambda - \mu = 1$, one can set $H_{r, l} := 3^{l-1 + L \cdot (r-1)} H$ to obtain
\begin{align*}
     \Omega = \Z_{\frac{3^{R\cdot L}-1}{2}}.
 \end{align*}
\end{itemize}
\end{theorem}

\section{Arbitrary Dimensional Sub-Generators}
\label{sec:k-dimensional generators}
\subsection{Golomb Ruler}
So far, we only considered QNNs under the constraint that all subgenerators are $2$-dimensional, i.e. Hermitian operators acting on a single qubit. In \cite{Peters2022GeneralizationDO, Kordzanganeh2022AnEF}, this was extended to single layer QNNs with a $d=2^R$-dimensional generator, which is nothing other than allowing an arbitrary data-encoding $H$. In this case, the the frequency spectrum $\Omega$ is maximal in size with 
\begin{align}
    |\Omega| = 2 \binom{d}{2} + 1 = 2^R\left(2^R-1\right) + 1
\end{align}
if and only if the eigenvalues of $H$ are a so called Golomb ruler. We extend this approach to QNNs with $k=2^q$-dimensional subgenerators and arbitrary many layers.
\begin{definition}[Golomb Ruler]
    Let $\lambda_1,...,\lambda_k \in \R$ with $\lambda_1\leq...\leq\lambda_k$. They are a \textbf{Golomb ruler} if all differences of pairs $\lambda_i-\lambda_j$ for $i \neq j$ are pairwise different. Equivalently, the set of differences has size
    \begin{align}
        \big|\Delta \lset{\lambda_1,...,\lambda_k}\big| = 2 \binom{k}{2} + 1.
    \end{align}
    The order of a Golomb ruler $G= \lset{\lambda_1,...,\lambda_k}$ is defined as 
    \begin{align}
        \ord(G) := k,
    \end{align}
    and its length as
    \begin{align}
        \ell(G) := \lambda_k - \lambda_1.
    \end{align}
    An \textbf{optimal Golomb ruler} is one with a minimal length over all Golomb rulers with the same order. 
     A Golomb ruler is called \textbf{perfect} if $\Z_{\ell(G)} = \Delta G$, i.e. there are no gaps in $\Delta G$. 
\end{definition}
Golomb rulers first appeared in \cite{Sidon1932EinS} and have applications in astronomy \cite{Blum1975SomeNP}, radio engineering \cite{Babcock1953IntermodulationII, Atkinson1986IntegerSW} and information theory \cite{Robinson1967ACO}. 

Golomb rulers are linked to the maximum frequency spectrum in size. This is precisely the notion needed to avoid degeneracy, i.e. all frequencies except 0 occur only once. We explain how to construct a single layer QNN using only $k$-dimensional generators with a maximum frequency spectrum. First, we need that the number of qubits $q$ per generator divides $R$, otherwise it would be impossible to use only $k$-dimensional generators with $R$ qubits. Next, find a Golomb ruler $G=\lset{\lambda_1,...,\lambda_k}$ of order $k=2^q$ and let $H \in \End{\B^{\otimes q}}$ be some Hermitian with eigenvalues $\lambda_1,...,\lambda_k$. Then scale $H$ by some appropriate factors $\beta_r \in \R$ to obtain the final generators $H_r := \beta_r \cdot H$. The key idea is to find a scaling such that the resulting frequency spectrum remains non-degenerate.
However, choosing arbitrary Golomb rulers for $H$ and large factors $\beta_r$ will lead to gaps in the resulting frequency spectrum $\Omega$. For approximation properties, as in \cite{Schuld2020EffectOD}, one does not only want a large set $\Omega$, but the largest possible $K\in \N$ such that $\Z_K \subseteq \Omega$, or at least as few gaps as possible. There are two ways to reduce the number of gaps. First, use an optimal Golomb ruler. If the Golomb ruler is even perfect, there are no gaps and both notions of maximality are equivalent. However, it can be shown that there is no perfect Golomb ruler of order $k\geq 5$ \cite{dimitromanolakis2002analysis} and finding optimal Golomb rulers of large orders is a difficult task and sometimes conjectured to be NP-hard \cite{Duxbury2021ACO}. Second, choose the factors $\beta_r$ as small as possible.

We would like to emphasise that in this work we do not investigate whether such non-separable generators $H_r$ can be easily implemented on real physical hardware.
This being said, the maximum frequency spectrum with respect to size is as follows.
\begin{theorem}[$|\Omega|$-Maximal Frequency Spectrum]
\label{thm: maximality with Golomb Ruler}
    The maximal frequency spectrum $\Omega_{\max}$ of a univariate QNN with $k=2^q$-dimensional generators $H_{r} \in \End{\mathcal{B}^{\otimes q}}$, where $q$ divides $R$, has size
    \begin{align*}
        |\Omega_{\max}| = \left(4^q - 2^q + 1\right)^{\nicefrac{RL}{q}}.
    \end{align*}
    In particular, given a fixed Hermitian $H\in \End{\mathcal{B}^{\otimes q}}$ such that $\sigma (H) \subseteq \Z + c$ for some constant $c \in \R$, choose any integer $\beta \geq 2 \ell(\sigma(H)) + 1$ and set $\beta_{r, l} := \beta^{l-1 + L\cdot (r-1)} \in \N$, $H_{r,l}:= \beta_{r,l} \cdot H$. This is defined as the \textbf{Golomb encoding}. The frequency spectrum of the QNN with this Golomb encoding is then maximal in size if and only if $\sigma(H)$ is a Golomb ruler.
\end{theorem}
A proof can be found in Appendix \ref{appendix: maximality proofs arbitrary dim}.
Note that we can recover the results of the previous section for $k=2$. In this case $q=1$ and therefore $|\Omega_{\max}| = 3^{RL}$. Moreover, since $k<5$, using a perfect Golomb ruler and minimal $\beta_r$ as in the proof of Theorem \ref{thm: maximality with Golomb Ruler}, there are no gaps in the frequency spectrum, thus obtaining the results of Theorem \ref{theorem: maximal frequency spektrum, arbitrary S}. If $q=R$ and therefore $k=d$, the result matches that in \cite{Peters2022GeneralizationDO} for single layer models.

\subsection{Relaxed Turnpike Problem}
Although a Golomb ruler is the right concept to maximise the size of the frequency spectrum, it is unsuitable for maximising $K \in \N$ such that $\Z_K \subseteq \Omega$. However, for the universality result in \cite{Schuld2020EffectOD}, it was crucial to have a wide range of integers without gaps in the frequency spectrum, hence such a maximally large $K$. Technically speaking, whereas in the Golomb ruler setting we were looking for a set $S$ such that $\Delta S$ has no degeneracy, here we are looking for a set $S\subseteq \Z$ of a given size $d$ such that $\Z_K \subseteq \Delta S$ for a maximally large $K\in \N$.
This task is somewhat related to what is known as the \textbf{turnpike problem} or the \textbf{partial digest problem} \cite{dakic2000turnpike}.
\begin{definition}[Turnpike Problem]
    Given a multiset $M$ of $\binom{d}{2}$ integers, find a set $S \subseteq \Z$ (of size $d$) such that $\Delta S = M$ (here, the difference is also understood as a multiset) if possible, else prove that there is no such set $S$.
\end{definition}
The turnpike problem occurs for example in DNA analysis \cite{Dix1988ErrorsBS}, X-ray crystallography \cite{Patterson1935ADM, Patterson1944AmbiguitiesIT} and other fields. Algorithms have been proposed to solve the turnpike problem based on, for example, backtracking \cite{skiena1990reconstructing} or factoring polynomials with integer coefficients \cite{lemke1988complexity}.

However, the turnpike problem is not precisely what we need here in our setting. For the frequency spectrum, the number of occurrences of a frequency is irrelevant, so here we only need M to be a set. Furthermore, we do not want to set $\Z_K = M$, but rather $ \Z_K \subseteq M$, hence relaxing some constraints. For this reason, we call this modification the \textbf{relaxed turnpike problem}.
\begin{definition}[Relaxed Turnpike Problem]
    Given $d \in \N$, find a set $S\subseteq \Z$ of size $d$ such that 
    \begin{align}
        K(S) := \max \lset{K \in \N_0|\ \Z_K \subseteq \Delta S}
    \end{align}
    is maximal, i.e. for all sets $S'\subseteq \Z$ of size $d$ one has $K(S') \leq K(S)$. $S$ is called a solution of the relaxed turnpike problem. Note that always $K(S) \leq \binom{d}{2}$. 
\end{definition}
To find a solution to the relaxed turnpike problem, we propose a not necessarily efficient algorithm. A proof can be found in Appendix \ref{appendix: maximality proofs arbitrary dim}.
\begin{theorem}[Algorithmic Solution of the Relaxed Turnpike Problem]
\label{theorem: algorithmic solution relaxed turnpike problem}
    For a given size $d\in \N$, define the finite set of candidates 
    \begin{align}
        C := \lset{\lset{s_1,s_2,...,s_d} \subseteq \Z|\ 0 < s_{i+1} - s_i \leq \binom{d}{2} \ \forall i=1,...,d-1, \ s_1 = 0}.
    \end{align}
    Then $C$ contains a solution to the relaxed turnpike problem. It can be found by iterating over all $S' \in C$ and returning the set $S$ with the largest $K(S')$.
\end{theorem}
Note that the size of the candidate set $C$ is $\binom{d}{2}^{d-1}$, making it impractical to iterate over all candidates, even for small $d$. For $d=1,\ldots, 8$ we computed all solutions of the relaxed turnpike problem, see Table \ref{tab: solution of the relaxed turnpike problem}.
\begin{table}[ht]
\centering
\begin{tabular}{c l c c r}
\toprule
\textbf{d} & \textbf{Example S}          & \textbf{K} & \textbf{Number of solutions}  & $\mathbf{|C|}$ \\ \toprule
1          & $\lset{0}$                  & 0          & 1    & 1                        \\ \midrule
2          & $\lset{0, 1}$               & 1          & 1    & 1                      \\ \midrule
3          & $\lset{0,1,3}$              & 3          & 2    & 9                    \\ \midrule
4          & $\lset{0,1,4,6}$            & 6          & 2    & 216                         \\ \midrule
5          & $\lset{0, 1, 2, 6, 9}$      & 9          & 8    & 10,000                     \\ \midrule
6          & $\lset{0, 1, 2, 6, 10, 13}$ & 13         & 14   & 759,375                     \\ 
\midrule
7          & $\lset{0, 2, 7, 13, 16, 17, 25}$ & 18         & 8 & 85,766,121
\\ 
\midrule
8          & $\lset{0, 8, 15, 17, 20, 21, 31, 39}$ & 24         & 2   & 13,492,928,512                     \\
\bottomrule
\end{tabular}
\caption{Example solution $S \in C$, $K=K(S)$ and number of solutions found in $C$ of the relaxed turnpike problem for $d =1,...,8$. The shown example is the solution with the lowest lexicographical order, defined from left two right.}
\label{tab: solution of the relaxed turnpike problem}
\end{table}
Note that in the context of QNNs, $d$ must be a power of $2$. For $d \leq 4$ the solutions are also perfect Golomb rulers, so the first case where the maximality in size and the maximality in $K$ differ is $d=8$.

With the previous thoughts we are prepared to return to the initial problem of searching for a QNN such that its frequency spectrum $\Omega$ contains $\Z_K$ with maximally large $K \in \N$. For a proof see Appendix \ref{appendix: maximality proofs arbitrary dim}.
\begin{theorem}[$K$-Maximally Frequency Spectrum]
\label{thm: maximality with turnpike}
    For a given dimension $d=2^R$, $K\in \N$ is maximal such that $\Z_K \subseteq \Omega$
    for some univariate single layer QNN with a $d$-dimensional generators $H \in \End{\BR}$ and frequency spectrum $\Omega$ if and only if the eigenvalues of $H$ are a solution of the relaxed turnpike problem.
\end{theorem}
The previous theorem cannot be extended in a trivial way to QNNs with $k$-dimensional generators and an arbitrary number of layers. By scaling we can construct a QNN such that $\Z_{K'} \subseteq \Omega$ with $K'=\nicefrac{((2K+1)^{\nicefrac{L \cdot R}{q}}-1)}{2}$, where $K=K(S)$ and $S$ is a solution of the relaxed turnpike problem of size $k$, but it is not guaranteed that this $K'$ is maximal. This is because there may be smaller scalings, such that the scaled spectra do not have to be disjoint, but would fill in gaps between each other and thus generate a larger $K'$. However, the scaling method gives us a further construction method for an encoding, which yields a frequency spectrum with a large $K'$, such that $\Z_{K'} \subseteq \Omega$. 
\begin{theorem}[Extension to more General Settings]
\label{thm: maximality with turnpike arbitrary k and L}
    Let $k=2^q$ with $q|R$ and let $H \in \End{\B^{\otimes q}}$ be Hermitian, such that $\sigma(H)$ is a solution of the relaxed turnpike problem of size $k$. Further, let $K:= K(\sigma(H))$, $\beta_{r, l} := (2K+1)^{l-1 + L\cdot(r-1)}$ and $H_{r,l} := \beta_{r,l} \cdot H$. This is defined as the \textbf{turnpike encoding}. Then
    \begin{align}
        \Z_{\frac{(2K+1)^{\nicefrac{R \cdot L}{q}}-1}{2}}\subseteq \Omega.
    \end{align}
    for the frequency spectrum $\Omega$ of the QNN with the $k$-dimensional generators $H_r$.
\end{theorem}
Note that the turnpike encoding can be used for arbitrary Hamiltonians $H$ without requiring $\sigma(H)$ to be a solution to the relaxed turnpike problem. However, in this case, maximality is no longer guaranteed.
The proof is again given in Appendix \ref{appendix: maximality proofs arbitrary dim}.
\section{Numerical Examples} 
\label{sec:examples}
We present numerical examples illustrating the frequency spectra associated with different encoding schemes for QNNs.
In Table~\ref{tab: numerical examples frequency spectra}, we report the size of the frequency spectrum $\Omega$ and the largest integer $K$ for which $\mathbb{Z}_K \subseteq \Omega$, for various encodings and small circuit shapes.
\begin{table}[!htbp]
\begin{tabular}{ll|rr|rr|rr|rr}
\toprule
\multicolumn{1}{c}{} & \multicolumn{1}{c}{} & \multicolumn{2}{c}{(3, 1)} & \multicolumn{2}{c}{(2, 2)} & \multicolumn{2}{c}{(3, 2)} & \multicolumn{2}{c}{(4, 2)} \\
Encoding & H & $|\Omega|$ & $K$ & $|\Omega|$ & $K$ & $|\Omega|$ & $K$ & $|\Omega|$ & $K$ \\
\midrule
Hamming & $\nicefrac{P}{2}$ & 7 & 3 & 9 & 4 & 13 & 6 & 17 & 8 \\
Binary & $\nicefrac{P}{2}$ & 15 & 7 & 31 & 15 & 127 & 63 & 511 & 255 \\
Ternary & $\nicefrac{P}{2}$ & 27 & 13 & 81 & 40 & 729 & 364 & 6561 & 3280 \\
Equal Layers & $\nicefrac{P}{2}$ & 27 & 13 & 25 & 12 & 125 & 62 & 625 & 312 \\
Golomb & $S_4$ & - & - & 169 & 84 & - & - & 28561 & 14280 \\
 & $G_8$ & 57 & 15 & - & - & 3249 & 15 & - & - \\
Turnpike & $S_4$ & - & - & 169 & 84 & - & - & 28561 & 14280 \\
 & $T_8$ & 53 & 24 & - & - & 2617 & 1200 & - & - \\
\bottomrule
\end{tabular}
\caption{Numerical examples of the frequency spectrum size $|\Omega|$ and the maximal integer $K \in \mathbb{N}_0$ such that $\mathbb{Z}_K \subseteq \Omega$ for different encodings and circuit shapes. Here, $P \in \{X,Y,Z\}$ denotes an arbitrary Pauli matrix. The sub-generators are given by $S_4 := \diag(0,1,4,6)$, $G_8 := \diag(0,1,4,9,15,22,32,34)$, and $T_8 := \diag(0,8,15,17,20,21,31,39)$. If $R \equiv 0 \ (\mathrm{mod}\ 2)$, we use $S_4$ as the sub-generator for both the Golomb and turnpike encodings; otherwise, we use $G_8$ for the Golomb encoding and $T_8$ for the turnpike encoding.
}
\label{tab: numerical examples frequency spectra}
\end{table}

For $R \equiv 0 \mod 3$, we employed distinct sub-generators for Golomb and turnpike encoding, 
\begin{align}
    G_8:=\diag(0, 1, 4, 9, 15, 22, 32, 34) \text{ and } T_8:=\diag(0, 8, 15, 17, 20, 21, 31, 39).
\end{align}
The spectrum of $G_8$ constitutes a Golomb ruler, whereas that of $T_8$ is a solution of the relaxed turnpike problem. Both operators act on three qubits, providing the smallest instance in which Golomb and turnpike encodings diverge. As expected, the Golomb encoding yield a larger frequency spectrum in size, while the number $K$ such that $\Z_K \subseteq \Omega$ is much larger for the turnpike encoding.

The distribution of the frequencies for specific shapes and encodings can be found in Figure \ref{fig:spectrum and degenerecies}.
\begin{figure}[ht]
    \centering
    \includegraphics[width=1.0\linewidth]{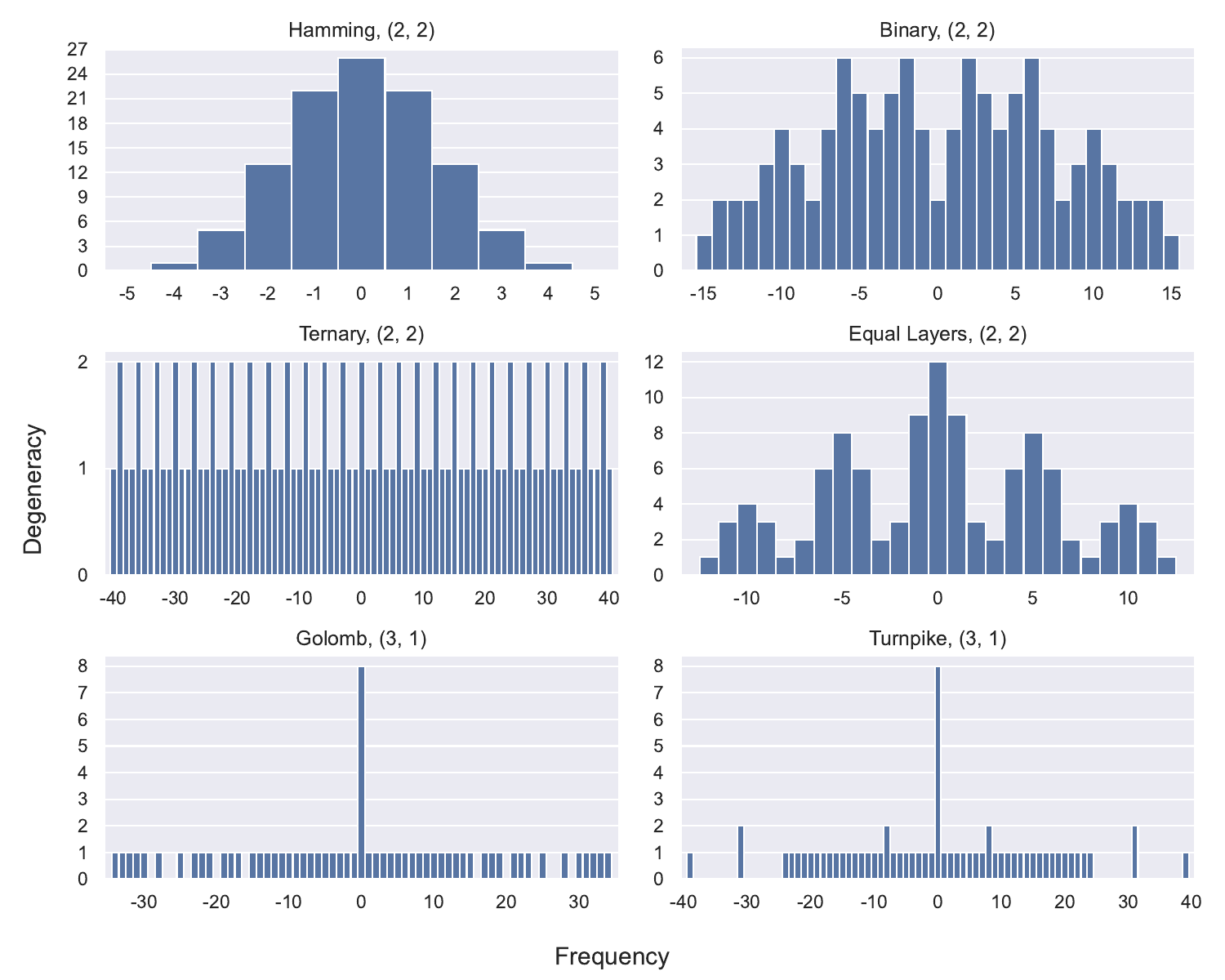}
    \caption{Frequency distributions for the considered QNN encoding schemes. The degeneracy indicates the number of distinct parameter combinations giving rise to the same frequency, i.e., the number of Fourier terms sharing that frequency. As in Table~\ref{tab: numerical examples frequency spectra}, we use $H = G_8$ for the Golomb encoding and $H = T_8$ for the turnpike encoding; for all other encodings, we set $H = \nicefrac{Z}{2}$. The circuit shape is $(2,2)$ for the Hamming, binary, ternary, and equal-layer encodings, and $(3,1)$ for the Golomb and turnpike encodings.
}
    \label{fig:spectrum and degenerecies}
\end{figure}
Whereas Hamming, binary, and equal-layer encodings exhibit narrower frequency spectra with higher degeneracy, ternary, Golomb, and turnpike encodings yield broader spectra with negligible degeneracy. Increased degeneracy generally provides greater flexibility, enabling the QNN to approximate specific frequencies more accurately, at the cost of reduced accessibility to higher frequencies.

We demonstrate the approximation quality for each encoding on a small toy example first. We fitted the quantum models of shape $(3, 1)$ on the target function 
\begin{align}
    g(x) = \frac{1}{20}\sum_{n=1}^9 \sin(n \cdot x)
\end{align}
for each encoding. We scaled the target function such that $g(x) \in [-\nicefrac{1}{2},\nicefrac{1}{2}]$, as models with Pauli strings as observables only can take values in the interval $[-1,1]$ and fitting them to target functions taking values on the boundaries $g(x) = \pm 1$ seems to be quite hard.
For the experiments, we implemented the QNNs with Python 3.12.7 and mainly the JAX \cite{jax} and pennylane \cite{pennylane} libraries. We discretized the interval $[-\pi, \pi]$ into 1000 equidistant points on which $g(x)$ was evaluated. We then used this as our training targets with \textit{Mean Squared Error (MSE)} as loss function. For optimization, we used the ADAM optimizer \cite{adam2014method} with $b_1=0.9$, $b_2=0.999$, $\epsilon = 10^{-8}$ and a learning rate of $\eta = 10^{-4}$. The results of the noise-free simulation can be found in Figure \ref{fig:approximation demonstration}.
\begin{figure}[ht]
    \centering
    \includegraphics[width=1\linewidth]{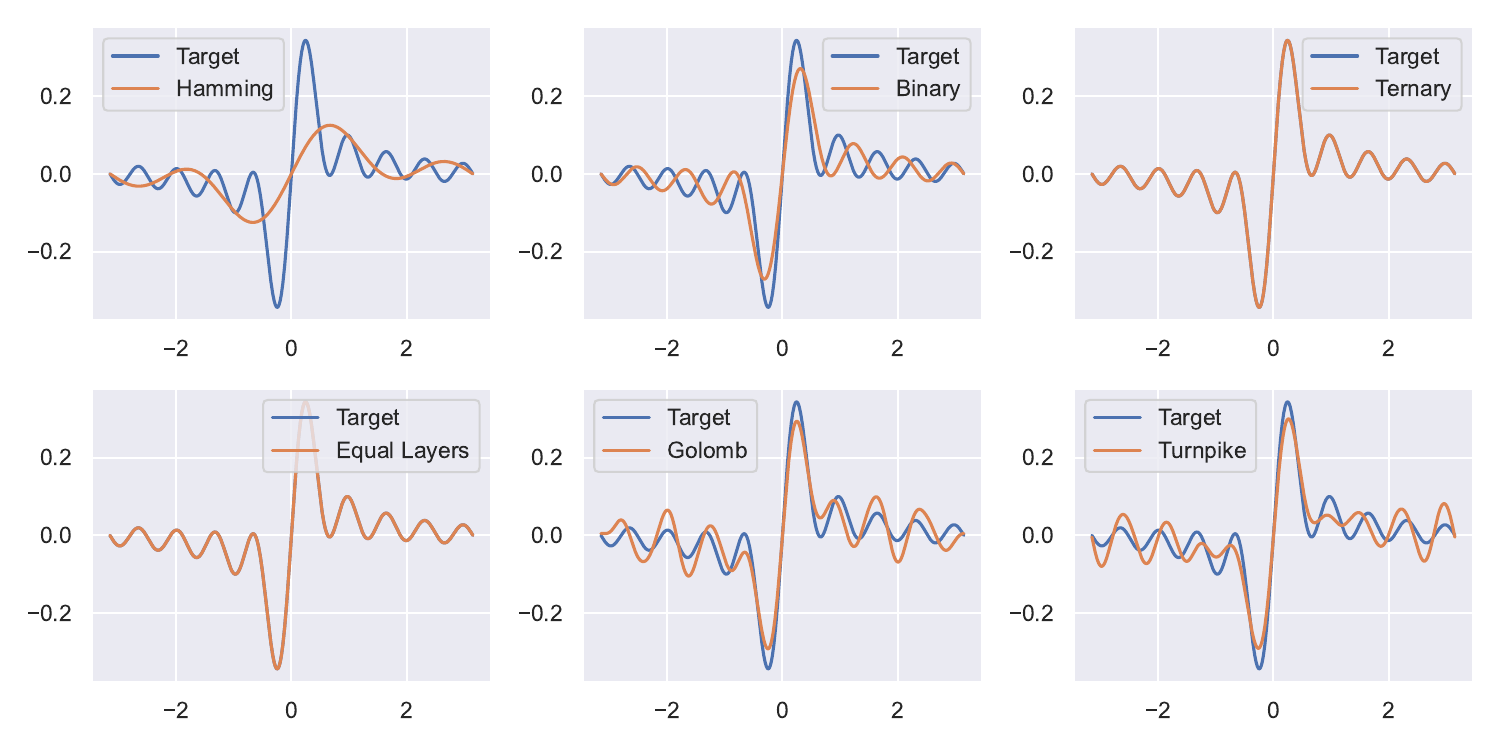}
    \caption{Approximation of the target function $g(x) = \frac{1}{20}\sum_{n=1}^9 \sin(n x)$ using different encoding schemes. All models have circuit shape $(3,1)$. The sub-generator $H$ is chosen as $\nicefrac{Z}{2}$ for the Hamming, binary, ternary, and equal-layer encodings, $G_8$ for the Golomb encoding, and $T_8$ for the turnpike encoding. The interval $[-\pi,\pi]$ is discretized into 1000 equidistant points, which are used as training targets. Model training is performed using ADAM with $b_1=0.9$, $b_2=0.999$, and $\epsilon = 10^{-8}$, treating the full dataset as a single batch per epoch. The encoding unitaries $W_{\boldsymbol{\theta}}^{(l)}$ are implemented using Pennylane’s \textit{StronglyEntanglingLayers} with 10 entangling layers. All models are trained for 10,000 epochs with learning rate $10^{-4}$ and mean squared error loss.
}
    \label{fig:approximation demonstration}
\end{figure}

According to Table \ref{tab: numerical examples frequency spectra}, the frequency spectra of Hamming and binary encoding for shape $(3,1)$ are too small to cover all frequencies, which can also be seen from the results in the figure. Since only a single layer is used, ternary and equal-layer encodings coincide. Both reproduce the target function exactly and are therefore capable of recovering its full frequency spectrum. While the frequency spectrum of Golomb and turnpike encodings is large enough in theory, the model apparently has problems approximating the objective function and thus all contained frequencies sufficiently well. Here, the large frequency spectrum apparently comes at the expense of approximation quality.
This is consistent with observations from \cite{Mhiri:2024abv} and an example of, what the authors call, frequency redundancy in the Fourier series spectrum.

We further test the impact of the specific shape $(R, L)$ at a fixed area $A = R \cdot L$ for different target functions. For that, again we discretized the interval $[-\pi, \pi]$ into $1000$ equidistant points and evaluated different target functions $g_i(x)$ on it. As target functions we have chosen
\begin{align}
    g_1(x) & = \frac{\left(-49x^{4} + 9x^{3} + 360x^{2} + 486\right)}{2000}, &
    g_2(x) &= \frac{1}{20}\sum_{n=1}^9 \sin(n \cdot x), & \\
    g_3(x) &= \frac{1}{10}\text{ReLU}(x) \text{ and} &
    g_4(x) &= \frac{\text{sign}(x)+1}{4}, &
\end{align}
where $ReLU(x)=\max\{0, x\}$ and $\text{sign}(x) = \nicefrac{x}{|x|}$ for $x \neq 0$ and $\text{sign}(0) = 0$.
Due to the implementation of pennylanes \textit{StronglyEntanglingLayers}, a model has
\begin{align}
    \# \text{Param} = 3 \cdot R \cdot (L+1) \cdot n_{\text{entangling layers}} = 3\cdot (A + R)\cdot n_{\text{entangling layers}}
\end{align}
many trainable parameters, where the factor $3$ arises from the implementation details of \textit{StronglyEntanglingLayers} and $n_{\text{entangling layers}}$ denotes the internal repetitions of layers of \textit{StronglyEntanglingLayers} and is therefore a hyperparameter we have to set. The formula is slightly asymmetric in $R$ and $L$ since there are always $L+1$ many parameter encoding unitaries in a QNN and the number is therefore not only dependent on the area $A$. To have a fair comparison for trainability, we have to choose the hyperparameter $n_{\text{entangling layers}}$ for each shape $(R,L)$ individually such that each model has roughly the same number of trainable parameters. We fix $A=6$ and determine $n_{\text{entangling layers}}$ for each $R$ according to Table \ref{tab:number entangling layers experiment}. 
\begin{table}[ht]
    \centering
    \begin{tabular}{c|c|c}
        $R$ \quad & $n_{\text{entangling layers}}$ & Number of trainable parameters\\
        \hline
        1 & 7 & 147\\
        2 & 6 & 144\\
        3 & 5 & 135\\
        6 & 4 & 144\\
    \end{tabular}
    \caption{Chosen number of entangling layers and resulting number of trainable parameters for each $R$ in the experiments for testing the trainability of models with the fixed area $A=6$.}
    \label{tab:number entangling layers experiment}
\end{table}
Again, we trained the models using $\nicefrac{Z}{2}$ as the sub-generator for 10,000 epochs with the ADAM optimizer and a learning rate of $\eta = 0.0001$ for the Hamming, binary, and ternary encodings. As loss function we used the \textit{Mean Squared Error (MSE)}. We omitted the equal-layer encoding in this experiment, as the property of having identical data-encoding layers is not necessarily preserved under area-preserving transformations. The Golomb and turnpike encodings were also excluded, since they would only be applicable to the shapes $(3, 2)$ and $(6, 1)$ for $k = 3$, and to $(2, 3)$ and $(6, 1)$ for $k = 2$, providing limited numerical insights for the trainability of models with different shapes but same area. We also attempted training with $A = 12$, but models with this larger area frequently
failed to achieve satisfactory loss values during training. To reduce the influence of randomness, we trained 10 models for each encoding and shape pair. The results shown in Figure \ref{fig:experiment results spectral invariance best run} correspond to the run (out of 10) with the smallest final MSE loss for each experimental configuration. In Figure \ref{fig:experiment results spectral invariance boxplot} we summarize the results of all 10 runs for each encoding, shape and function in a boxplot.
\begin{figure}
    \centering
    \includegraphics[width=1\linewidth]{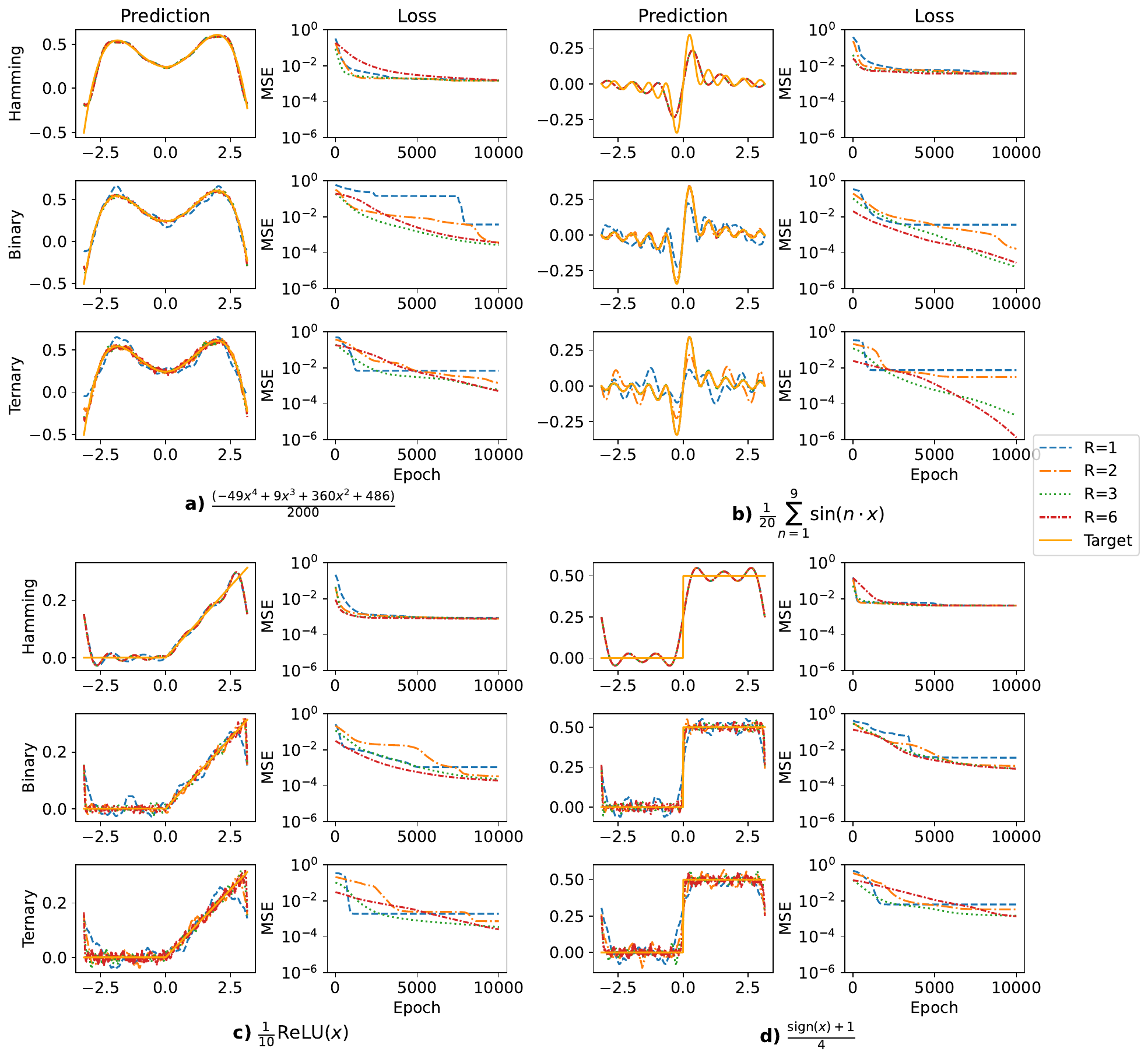}
    \caption{Experimental results for all four tested target functions $g_i(x)$. The input domain $[-\pi, \pi]$ was discretized into 1000 equidistant points, and the models were trained using the MSE loss for 10,000 epochs with a learning rate of $\eta = 0.0001$. The area was fixed to $A = R \cdot L = 6$, and $n_{\text{entangling layers}}$ was adjusted for each shape $(R, L)$ to ensure comparable numbers of trainable parameters. We tested Hamming, binary, and ternary encodings for $R = 1, 2, 3, 6$. For each configuration and model, ten training runs were performed, and the model with the lowest final loss was selected. Shown are the resulting predictions and corresponding loss curves for each target function.
}
    \label{fig:experiment results spectral invariance best run}
\end{figure}
\begin{figure}
    \centering
    \includegraphics[width=1\linewidth]{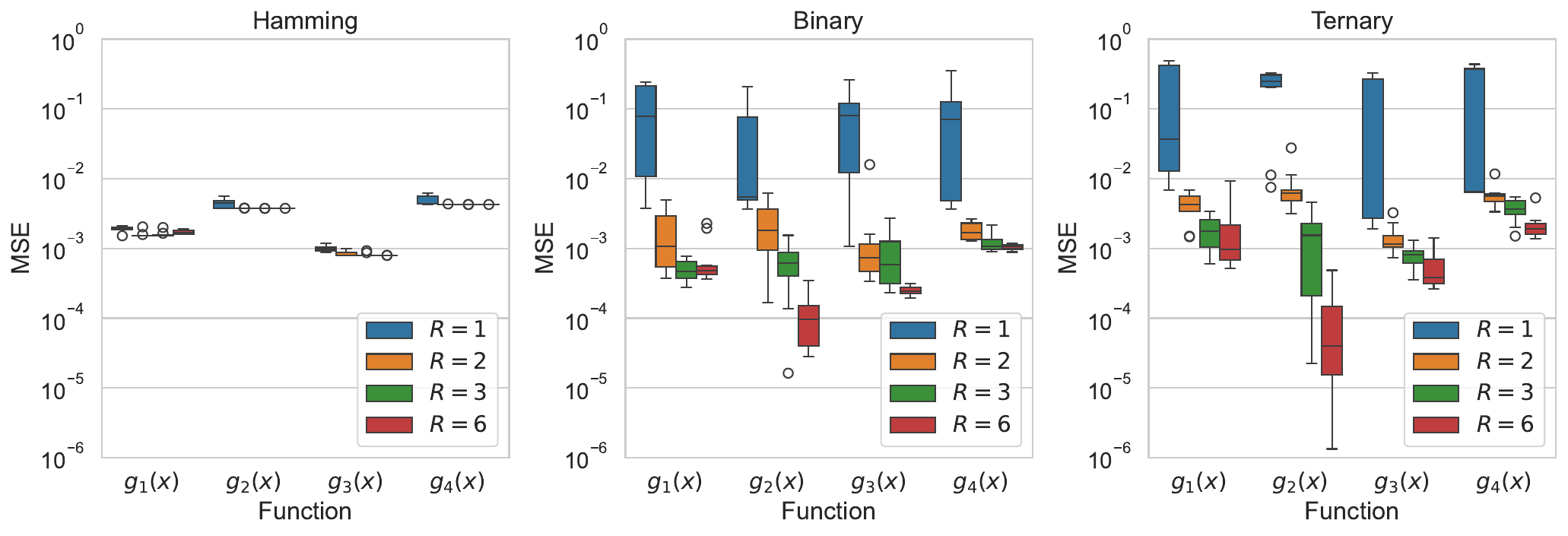}
    \caption{Boxplot summary of the final MSE losses over 10 independent training runs for each target function $g_i(x)$, encoding scheme (Hamming, binary, ternary), and shape $(R, L)$ with fixed area $A = 6$. Each boxplot displays the distribution of final MSE values after 10,000 epochs of training. Lower MSE values indicate better approximation of the target function.
}
    \label{fig:experiment results spectral invariance boxplot}
\end{figure}

For each shape $(R, L)$ and each encoding, the best-performing model is able to approximate the respective target function to a certain extent, as illustrated in Figure \ref{fig:experiment results spectral invariance best run}. However, consistently comparable performance across all shapes $(R, L)$ is only observed for the Hamming encoding. In contrast, for the binary and ternary encodings, the training behavior and convergence quality depend noticeably on the specific choice of $(R, L)$.

In particular, for smaller values of $R$, the optimization tends to stagnate at higher loss levels compared to configurations with larger $R$. In some cases, for example for the target function $g_1(x)$ with binary encoding, the optimization process initially appears to stagnate, subsequently improves after a certain number of iterations, and then stabilizes again at a lower loss level. For configurations with $R \geq 3$, the loss curves are generally smoother and exhibit a more consistent decrease over the course of training.

While Figure \ref{fig:experiment results spectral invariance best run} displays only the best-performing model, i.e., the model achieving the lowest final loss for each configuration, Figure \ref{fig:experiment results spectral invariance boxplot} summarizes the results over all $10$ training runs per configuration. These aggregated results further support the observation that, for the Hamming encoding, the approximation quality is largely independent of the specific model shape. In contrast, models employing binary or ternary encodings tend to achieve better overall performance for larger values of $R$ across all considered target functions.

In summary, a noticeable difference in the trainability of the models can be observed depending on the choice of the shape $(R, L)$. While the frequency spectrum of the models depends only on the area $A$ and not on the specific choice of the shape $(R, L)$, at least for a fixed encoding, the shape nevertheless appears to influence the trainability of the model, except for the encoding with the smallest frequency spectrum among those considered, namely the Hamming encoding. In our experiments, larger values of $R$ tend to yield more favorable optimization behavior than smaller ones. This does not contradict the results from Section \ref{sec: Spectral Invariance Under Area-Preserving Transformations}, since that section concerns only the frequency spectrum and not the trainability or universal equivalence of the models and their corresponding Fourier series.

The practical relevance of these findings lies primarily in improving our understanding of the structure of the frequency spectrum. Models with an encoding that produces a large frequency spectrum per area require fewer resources, meaning a smaller number of qubits and layers and therefore a smaller area $A$, than encodings with smaller frequency spectra per area. However, our experiments suggest that a large frequency spectrum can also make it more difficult for the models to approximate the target function accurately or to train effectively, even if the theoretical expressivity is sufficient. Therefore, in practice, one must empirically balance whether, for a given task, models with a smaller area and an encoding leading to a larger frequency spectrum per area, or conversely models with a larger area and higher resource consumption but a slightly smaller frequency spectrum per area, are more suitable. Maximizing the frequency spectrum is therefore not an objective in itself, rather, it delineates the boundaries within which an appropriate model selection must take place.

\section{Conclusion and Outlook} 
\label{sec:chapter_conclusion}
In this work, we have examined the properties of the frequency spectrum of QNNs in more detail. We focused on models whose approach consists of alternating parameterized encoding layers.

First, we discussed the representability of QNNs by finite Fourier series, which is a well established fact. It was observed that the frequency spectrum of the multivariate model is the Cartesian product of the frequency spectrum of univariate models. Therefore, it was sufficient to study only univariate models.
We then introduced the concept of spectral invariance under area-preserving transformations. This states that the frequency spectrum of a QNN is invariant under certain transformations of the ansatz that leave the area $A=RL$ unchanged, where $R$ is the number of qubits and $L$ the number of layers. The consequence of this is that, from a frequency spectrum point of view, the frequency spectrum depends only on the area $A=RL$ and not on the individual $R$ and $L$ without further restrictions to the layers. In addition, one can consider single-layer or single-qubit models and extend frequency-spectrum results to QNNs of arbitrary shape. This explains various symmetry results in $R$ and $L$ in the literature.

Additionally, we have comprehensively clarified how the generators of a QNN must be chosen in order to obtain a maximal frequency spectrum for a given area $A$. For that, we have distinguished two cases of maximality, namely maximal in the size $|\Omega|$ of the frequency and maximal in $K$ such that $\Z_K \subseteq \Omega$.
To maximize the size $|\Omega|$ of the frequency spectrum, we extended existing results to QNNs with arbitrary number of layers and arbitrary dimensional generators based on the concept of a Golomb ruler.
For maximality in $K$, we introduced the relaxed turnpike problem, which is a variation of the classical turnpike problem, and proved that a univariate single layer QNN with one generator is maximal if and only if the eigenvalues are a solution to the relaxed turnpike problem. 

Our work characterizes and extends existing results on the maximal frequency spectrum attainable by QNNs, thereby providing a theoretical bound on the class of target functions that can be approximated for a given model architecture. While this does not yet address how well a target function within the approximable class can actually be represented, particularly with respect to the choice of encoding, our systematic approach and notion establishes a foundation for a deeper theoretical understanding of QNNs as models for quantum machine learning.

These findings have broader implications for Quantum Machine Learning. The invariance of the frequency spectrum with respect to the area 
$A=R \cdot L$ suggests that resources in QNN design can be traded more flexibly than commonly assumed: from a spectral standpoint, adding layers or adding qubits contributes equally to enlarging the frequency spectrum. This observation is particularly relevant in the NISQ era, where the availability of qubits is often severely limited while circuit depth is constrained by noise. Spectral invariance indicates that, under certain conditions, one can compensate a lack of qubits with deeper circuits or vice versa, without sacrificing expressible frequencies. 

Our experiments, however, show that the shape does indeed have an influence on the trainability of the models. While a smaller number of qubits appears to be advantageous for simple target functions, models tend to become trapped in plateaus more frequently during the training of more complex target functions when fewer qubits are used, at least for encodings with a large frequency spectrum per area. This demonstrates that the shape can have a tangible effect on model training, although the nature of this influence depends on the complexity of the target function. Furthermore, our experiments indicate that, in practice, a trade-off must be made between encodings with a relatively small frequency spectrum per area, requiring a larger area and thus higher resource consumption, and encodings with a large frequency spectrum per area but smaller overall area $A$.

Naturally, our study also has limitations. We have focused exclusively on the theory of the frequency spectrum of Quantum Neural Networks, leaving aside the distribution of Fourier coefficients and their impact on optimization landscapes. While we have examined the practical expressivity and trainability of the models using several examples, a systematic understanding of the influence of the shape and the encodings on the training process remains an open question. Furthermore, we expect that our results extend to other classes of models based on parameterized quantum circuits, but this was beyond the scope of the present work and should be investigated in future research.

Although the focus was solely on the frequency spectrum, future work could expand the area-preserving transformations to include a transformation of the parameter encoding layers $W^{(l)}$ and the observable $M$. This would allow to study their impact on the entire Fourier series and could provide insights into the expressiveness and trainability of QNNs. An extension to other architectures of parametrized circuits could also be fruitful as a next step for further insights.

\backmatter

\section*{Declarations}
\subsection*{Funding}
This work was partially funded by the BMWK project EniQmA (01MQ22007A).
\subsection*{Data Availability}
Not applicable.
\subsection*{Author Contribution}
Patrick Holzer and Ivica Turkalj contributed equally to this work.
\subsection*{Competing Interests}
The authors declare no competing interests.

\begin{appendices}

\section{Frequency Spectrum of Multivariate QNNs}
\label{appendix:frequency spectrum of multivariate models}
First we give a proof of Theorem \ref{theorem: Univariate QNN is Fourier} for the frequency spectrum of the univariate Model.
\begin{proof}
    For each $l=1,\ldots,L$ let $\lambda_0^{(l)}, ..., \lambda_{d-1}^{(l)}$ be the eigenvalues of $H_l$.
    The action of $W^{(l)}$ and $ e^{-ix H_l}$ on some basis vector $|j\rangle$ for $j \in [d]$ are given by
    \begin{align*}
        W^{(l)}|j\rangle = \sum_{i=0}^{d-1} W_{i, j}^{(l)} |i\rangle
    \end{align*}
    and
    \begin{align*}
       S_l(x)|j\rangle = e^{-ix H_l}|j\rangle = e^{-ix \lambda_j^{(l)}} |j\rangle,
    \end{align*}
    where $ W_{i, j}^{(l)} := \langle i |  W^{(l)}|j\rangle \in \C$.
    Hence
    \begin{align*}
        U(x) |0\rangle &= W^{(L+1)}S_L(x) \cdots S_1(x) W^{(1)} | 0\rangle \\
        &=\sum_{j_1=0}^{d-1} W_{j_1, 0}^{(1)} \left( W^{(L+1)}S_L(x) \cdots S_1(x) |j_1 \rangle\right) \\
        &=\sum_{j_1=0}^{d-1}\sum_{j_2=0}^{d-1} W_{j_2, j_1}^{(2)} W_{j_1, 0}^{(1)} \cdot e^{-ix \lambda_{j_1}^{(1)}}\left( W^{(L+1)}S_L(x) \cdots S_2(x) |j_2 \rangle\right) \\
        &=\sum_{j_1=0}^{d-1} \sum_{j_2=0}^{d-1}  W_{j_2, j_1}^{(2)} W_{j_1, 0}^{(1)} \cdot e^{-ix (\lambda_{j_1}^{(1)}+\lambda_{j_2}^{(2)})} \left( W^{(L+1)}S_L(x)\cdots W^{(3)}|j_2 \rangle\right) \\
        &\vdots \\
        &= \sum_{j_{L+1}=0}^{d-1}\sum_{\bj \in [d]^{L}}\left( \prod_{l=1}^{L+1} W_{j_l, j_{l-1}}^{(l)} \right) e^{-ix \Lambda_{\bj}} |j_{L+1} \rangle,
    \end{align*}
    where $j_0 := 0$ and $\Lambda_{\bj} := \sum_{l=1}^L \lambda_{j_l}^{(l)} \in \sum_{l=1}^L \sigma (H_l)$. This yields
    \begin{align*}
        \langle 0| U^\dagger(x)  = \sum_{j_{L+1}=0}^{d-1} \langle j_{L+1}| \sum_{\bj \in [d]^{L}} \left(\prod_{l=1}^{L+1} \left(W^\dagger\right)_{j_{l-1}, j_l}^{(l)}\right) e^{ix \Lambda_{\bj}}.
    \end{align*}
    Hence
    \begin{align*}
        f(x) = \langle 0| U^\dagger(x) M U(x) |0\rangle  = \sum_{\bj, \bk \in [d]^{L}} a_{\bk, \bj }e^{ix (\Lambda_{\bk}-\Lambda_{\bj})}
    \end{align*}
    with 
    \begin{align*}
        a_{\bk, \bj} = \sum_{k_{L+1}, j_{L+1} = 0}^{d-1} \left(\prod_{l=1}^{L+1} W_{k_l, k_{l-1}}^{(l)}\right) \cdot \left(\prod_{l=1}^{L+1} \left(W^\dagger\right)_{j_{l-1}, j_l}^{(l)}\right) \cdot M_{j_{L+1}, k_{L+1}}
    \end{align*}
    and $\bk = (k_1,...,k_L), \bj = (j_1,...,j_L) \in [d]^L$.
    If we group all terms with the same frequencies together, we obtain
    \begin{align*}
        f(x) =\sum_{\omega \in \Omega} c_\omega e^{i\omega x}
    \end{align*}
    where
    \begin{align*}
        c_\omega :=  \sum_{\substack{\bj, \bk \in [d]^{L} \\ \Lambda_{\bk} - \Lambda_{\bj} = \omega}} a_{\bk, \bj }
    \end{align*}
    and
    \begin{align}
    \Omega = \lset{\Lambda_{\bk} - \Lambda_{\bj}|\bj, \bk \in [d]^L} = \Delta \sum_{l=1}^L \sigma (H_l).
\end{align}
By Lemma \ref{lemma:minkowski properties}, the difference $\Delta$ can be shifted inside the sum, proving the claim.
\end{proof}
We extend Theorem \ref{theorem: Univariate QNN is Fourier} to multivariate models and give a proof of Theorem \ref{theorem: frequency spectrum of multivariate QNN is rectangular}.
First, we begin with models with a parallel ansatz.
\begin{theorem}[Multivariate Frequency Spectrum - Parallel Ansatz]
\label{theorem: Multivariate QNN is Fourier - parallel ansatz}
Let  $f(x) = \langle 0 | U^\dagger (\bx) M U(\bx)|0\rangle$ be a multivariate QNN with a parallel ansatz.
Then
\begin{align}
    f(\bx) = \sum_{ \bomega \in \bOmega} c_{\bomega}  e^{-i \bomega \cdot \bx},
\end{align}
with
\begin{align*}
    \bOmega = \sum_{l=1}^L \Delta(\sigma(H_{l, 1})\times\ldots \times \sigma(H_{l, N})).
\end{align*}
\end{theorem}
\begin{proof}
    If $\lambda_0^{(l, n)}, ..., \lambda_{d-1}^{(l, n)}$ are the eigenvalues of $H_{l, n}$, we write
    \begin{align*}
        \blambda^{(l)}_{\bj}:= \left(\lambda_{j_1}^{(l, 1)},...,\lambda_{j_N}^{(l, N)}\right)^T \in \sigma(H_{l, 1})\times\ldots \times \sigma(H_{l, N}) \subseteq \R^N
    \end{align*}
    for all $\bj \in \left[d\right]^N$ (recall that $d=2^R$).
   To obtain the representation as a finite Fourier series, we again have to determine the action of $W^{(l)}$ and $S_l(\bx)$ on arbitrary vectors as in the proof of Theorem \ref{theorem: Univariate QNN is Fourier}.

    Let
    \begin{align*}
        |\bj\rangle := |j_1,...,j_N\rangle := \bigotimes_{n=1}^N |j_n\rangle \in \bigotimes_{n=1}^N \BR = \B^{\otimes R \cdot N}
    \end{align*}
    for all $\bj \in \left[d\right]^N$, where $|0\rangle,...,|d-1\rangle\in \BR$ denotes the computational basis.

    The action of $W^{(l)}$ and $S_l(\bx)$ on some basis vector $|\bj\rangle$ for $\bj \in \left[d\right]^N$ are given by
    \begin{align*}
        W^{(l)}|\bj\rangle = \sum_{\bi \in \left[d\right]^N} W_{\bi, \bj}^{(l)} |\bi\rangle
    \end{align*}
    and
    \begin{align*}
       S_l(\bx)|\bj\rangle &= \bigotimes_{n=1}^N e^{-ix_n H_{l, n}}|j_n\rangle \\
       &= e^{-i\sum_{n=1}^N x_n \lambda^{(l, n)}_{j_n}} |\bj\rangle\\
       &= e^{-i \blambda^{(l)}_{\bj} \cdot \bx} |\bj\rangle 
    \end{align*}
    where $ W_{\bi, \bj}^{(l)} := \langle \bi |  W^{(l)}|\bj\rangle \in \C$.
    We obtain the same equations as in the proof of Theorem \ref{theorem: Univariate QNN is Fourier}, we only have to replace each index $j$ by a multiindex $\bj$ and the complex product $\lambda_j^{(l)} \cdot x$ by the scalar product $\blambda_{\bj}^{(l)} \cdot \bx$.

    We therefore obtain
    \begin{align*}
        U(\bx) |0\rangle  = \sum_{\bj^{(1)}, ...,\bj^{(L+1)} \in \left[d\right]^N}\left( \prod_{l=1}^{L+1} W_{\bj^{(l)}, \bj^{(l-1)}}^{(l)} \right) e^{-i\bx \cdot \left(\sum_{l=1}^L\blambda_{\bj^{(l)}}^{(l)}\right)} |\bj^{(L+1)} \rangle,
    \end{align*}
    where $\bj^{(0)} := (0,...,0)$. 
    Hence
    \begin{align*}
        f(\bx) = \langle 0| U^\dagger(\bx) M U(\bx) |0\rangle  = \sum_{\substack{\bk^{(1)}, ...,\bk^{(L+1)} \in \left[d\right]^N \\
        \bj^{(1)}, ...,\bj^{(L+1)} \in \left[d\right]^N}} 
        a_{\bk^{(1)},...,\bk^{(L)}, \bj^{(1)},...,\bj^{(L)}} e^{i\bx \cdot  \left(\sum_{l=1}^L\left(\blambda_{\bk^{(l)}}^{(l)}-\blambda_{\bj^{(l)}}^{(l)}\right)\right)}
    \end{align*}
    with 
    \begin{align*}
        a_{\bk^{(1)},...,\bk^{(L)}, \bj^{(1)},...,\bj^{(L)}} = \sum_{\bk^{(L+1)}, \bj^{(L+1)} \in \left[d\right]^N} \left(\prod_{l=1}^{L+1} W_{\bk^{(l)},\bk^{(l-1)}}^{(l)}\right) \cdot \left(\prod_{l=1}^{L+1} \left(W^\dagger\right)_{\bj^{(l-1)}, \bj^{(l)}}^{(l)}\right) \cdot M_{\bj^{(L+1)}, \bk^{(L+1)}}.
    \end{align*}
    Again, grouping all terms with the same frequencies together yields
    \begin{align*}
        f(x) =\sum_{\bomega \in \bOmega} c_{\bomega} e^{i\bomega \cdot \bx}
    \end{align*}
    with
    \begin{align}
    \bOmega = \sum_{l=1}^L \Delta(\sigma(H_{l, 1})\times\ldots \times \sigma(H_{l, N})).
    \end{align}
\end{proof}
We now show that the parallel and the sequential ansatz lead to the same frequency spectrum.

\begin{theorem}[Multivariate Frequency Spectrum - Sequential Ansatz]
\label{theorem: Multivariate QNN is Fourier - sequential ansatz}
Let  $f(x) =  \langle 0 | U^\dagger (\bx) M U(\bx)|0\rangle$ be a multivariate QNN with a sequential ansatz.
Then 
\begin{align}
    f(\bx) = \sum_{\bomega  \in \bOmega} c_{\bomega} e^{-i \bomega \cdot \bx}
\end{align}
with the frequency spectrum $\bOmega$ being the same as in the corresponding parallel ansatz from Theorem \ref{theorem: Multivariate QNN is Fourier - parallel ansatz}.
\end{theorem}
\begin{proof}
   The proof works very similarly to Theorem \ref{theorem: Univariate QNN is Fourier}.
    We use the same notation for the eigenvalues of $H_{l,n}$ as in Theorem \ref{theorem: Multivariate QNN is Fourier - parallel ansatz}.
    The ansatz circuit is of the form
    \begin{align}
        U(\bx) = U_N(x_N) \cdots U_1(x_1).
    \end{align}
    In the proof of Theorem \ref{theorem: Univariate QNN is Fourier} we have seen that
    \begin{align}
        U_n(x_n)|j\rangle = \sum_{j_{L+1}\in [d]} \sum_{\bj \in [d]^L}\left(\prod_{l=1}^{L+1}W_{j_l, j_{l-1}}^{(l)}\right)e^{-ix_n \Lambda_{\bj}^{(n)}} |j_{L+1}\rangle
    \end{align}
    with $j_0:=j$ (we only considered $j=0$ in the proof of Theorem \ref{theorem: Univariate QNN is Fourier}) and $\Lambda_{\bj}^{(n)} := \sum_{l=1}^L \lambda_{j_l}^{(l, n)}$.
    Applying all $U_n(x_n)$ consecutively, $U(\bx)$ has the form
    \begin{align*}
        U(\bx) |0\rangle  
        = \sum_{\bj^{(1)}, ...,\bj^{(N+1)} \in \left[d\right]^L}c_{\bj^{(1)},...,\bj^{(N+1)}} e^{-i \sum_{n=1}^N x_n\Lambda_{\bj^{(n)}}^{(n)}} |j^{(L+1)}_L \rangle.  
    \end{align*}
    The rest follows analogously to Theorem \ref{theorem: Multivariate QNN is Fourier - parallel ansatz}.
\end{proof}
Finally, we proof Theorem \ref{theorem: frequency spectrum of multivariate QNN is rectangular} and show, that the frequency of the multivariate model is just the Cartesian product of the frequency spectra of the corresponding univariate models.
\begin{proof}
    This follows from Lemma \ref{lemma:minkowski properties} about the Minkowski sum.
    \begin{align*}
        \bOmega_{L, N}\left(\left(H_{l, n}\right)_{l, n}\right) 
        &=  \sum_{l=1}^L\Delta(\sigma(H_{l, 1})\times\ldots \times \sigma(H_{l, N})) \\
        &= \sum_{l=1}^L \Delta(\sigma(H_{l, 1}))\times\ldots \times \Delta(\sigma(H_{l, N})) \\
        &= \sum_{l=1}^L \Delta(\sigma(H_{l, 1}))\times\ldots \times \sum_{l=1}^L\Delta(\sigma(H_{l, N})))\\
        &= \Omega_{L}\left(H_{1, 1},...,H_{L, 1}\right) \times \cdots \times \Omega_{L}\left(H_{1, N},...,H_{L, N}\right).
    \end{align*}
\end{proof}

\section{Maximality Proofs for 2-Dimensional Sub-Generators}
\label{appendix: maximality proofs 2dim}
The following purely number theoretic lemma is necessary to prove the maximality results.
\begin{lemma}
\label{lemma: uniqueness of the solution}
    Let $R,L \in \N$ and $z_1,...,z_R \in \R_{\geq 0}$ with $z_1 \leq ... \leq z_R$ such that
    \begin{align}
    \label{eq: set_equation number theoretic lemma}
    \sum_{r=1}^R z_r \cdot \Z_L = \Z_{\frac{(2L+1)^R-1}{2}}.
    \end{align}
    The unique solution in $z_1,...,z_R$ is given by
    \begin{align*}
        z_r = (2L+1)^{r-1} 
    \end{align*}
for all $r = 1,...,R$.
\end{lemma}
\begin{proof}
   One can easily check that $z_r = (2L+1)^{r-1}$ for all $r = 1,...,R$ is indeed a solution of the problem. 

   To prove the uniqueness, we first show that $\sum_{r=1}^R z_r \cdot \Z_L = \Z_{\frac{(2L+1)^R-1}{2}}$ if and only if 
   \begin{align}
       \sum_{r=1}^R z_r \cdot [2L+1] = \left[(2L+1)^R\right].
   \end{align}
    Since all $z_1,...,z_R$ are non negative, one has
    $T:=\sum_{r=1}^R z_r \cdot L = \frac{(2L+1)^R-1}{2}$ as it is the maximal element of both sides.
   Then
   \begin{align}
      \sum_{r=1}^R z_r \cdot [2L+1] = \sum_{r=1}^R z_r \cdot (\Z_L + L) = \left(\sum_{r=1}^R z_r \Z_L \right) + T = \Z_{\frac{(2L+1)^R-1}{2}} + T = \left[(2L+1)^R\right].
   \end{align}
   The opposite direction is analogue. Hence we consider the equation $\sum_{r=1}^R z_r \cdot [2L+1] = \left[(2L+1)^R\right]$ from now on. 
   The set on the left hand side has maximally size
   \begin{align}
       \Big|\sum_{r=1}^R z_r \cdot [2L+1]\Big| \leq (2L+1)^R, 
   \end{align}
    thus it could only be equal $\left[(2L+1)^R\right]$ if all sums $\sum_{r=1}^R z_r \cdot s_r \in \sum_{r=1}^R z_r \cdot [2L+1]$ are pairwise different.
    Further, all $z_1,\ldots, z_R \in \N_0$, since
    \begin{align}
        z_r \in \sum_{r=1}^R z_r \cdot [2L+1] = \left[(2L+1)^R\right] \subseteq \N_0.
    \end{align}
   As a consequence $z_1 \geq 1$, else 
   \begin{align}
       0 = 0\cdot z_1 + \ldots + 0 \cdot z_R = 1 \cdot z_1 + 0 \cdot z_2 + \ldots 0 \cdot z_R
   \end{align}
   and the sums would therefore not be pairwise different. We can conclude further that $z_1 = 1$ from $1 \leq z_1 \leq \ldots \leq z_R$, else
   \begin{align}
       1 < z_1 \leq s_1 \cdot z_1 + \ldots + s_R \cdot z_R
   \end{align}
   if at least one coefficient is non zero, contradicting $1 \in \left[(2L+1)^R\right] = \sum_{r=1}^R z_r \cdot [2L+1]$.
   Since $z_1 = 1$ and $s = z_1 \cdot s$ for all $s \in [2L+1]$, we conclude that $z_2 \geq 2L+1$, again by uniqueness of the sums. By the same reasons, $z_2 > 2L+1$ would lead to a contradiction because then 
   \begin{align}
       2L+1 < z_2 \leq s_1 \cdot z_1 + \ldots + s_R \cdot z_R
   \end{align}
   if at least one of $s_2,\ldots s_R$ is non zero and therefore
   $2L+1 \not \in \sum_{r=1}^R z_r \cdot [2L+1] =  \left[(2L+1)^R\right]$.
   Repeating this argument, one concludes $z_r=(2L+1)^{r-1}$. 
\end{proof}

With the previous lemma we are able to prove Theorem \ref{theorem: maximal frequency spektrum equal data encoding layers}.
\begin{proof}
By Theorem \ref{theorem: Univariate QNN is Fourier} and Lemma \ref{lemma:basic properties kronecker sum}, the frequency spectrum $\Omega$ is given by
\begin{align}
    \Omega
        = \sum_{l=1}^L\sum_{r=1}^R\Delta \sigma(H_r) 
        = \sum_{l=1}^L\sum_{r=1}^R(\lambda_r-\mu_r) \cdot \Z_1 
        = \sum_{r=1}^R(\lambda_r-\mu_r) \cdot \Z_L 
\end{align}
    The upper bound can be derived from counting elements
    \begin{align*}
        |\Omega| 
        &\leq (2L+1)^R,
    \end{align*}
    which yields $K \leq \frac{(2L+1)^R -1}{2}$ by the symmetry $\Omega = - \Omega$ and $0 \in \Omega$. By defining $z_r := \lambda_r - \mu_r$ for all $r = 1,...,R$, the rest is the exact statement of Lemma \ref{lemma: uniqueness of the solution}.
\end{proof}
Given the proof for equal data encoding layers, we can extend the results to arbitrary encoding layers and prove Theorem \ref{theorem: maximal frequency spektrum, arbitrary S}.
\begin{proof}
    By Theorem \ref{theorem: Spectral Invariance Under Area-Preserving Transformations}, each QNN of shape $(R, L)$ with $2$-dimensional generators has the same frequency spectrum as a single layer model ($L'=1$) with $2$-dimensional generators and $R':= R\cdot L$ many qubits.
    For a single layer model, Theorem \ref{theorem: maximal frequency spektrum equal data encoding layers} 
    states that
    \begin{align*}
        K\leq \frac{(2L'+1)^{R'} -1}{2} = \frac{3^{R\cdot L}-1}{2}
    \end{align*}
    with equality if and only if $\lset{\lambda_r^{(l)} - \mu_r^{(l)}|\ r, l} = \lset{3^0,3^1,\ldots, 3^{R'-1}}$.
    Note that the bijection 
    \begin{align}
        [R]\times [L] & \rightarrow [RL] \\
        (r, l) & \mapsto l-1 + L \cdot (r-1)
    \end{align}
    is arbitrarily chosen and can be replaced by any other bijection.
\end{proof}

\section{Maximality Proofs for Arbitrary Dimensional Sub-Generators}
\label{appendix: maximality proofs arbitrary dim}
\subsection{Golomb Encoding}
We give a proof of Theorem \ref{thm: maximality with Golomb Ruler}.
\begin{proof}
By Theorem \ref{theorem: Spectral Invariance Under Area-Preserving Transformations} it is sufficient to prove the statement for single layer QNNs.
     The frequency spectrum of any univariate single layer QNN with no further restrictions to the data encoding layers and $k=2^q$-dimensional generators $H_{r} \in \End{\B^{\otimes q}}$ is given by
     \begin{align}
         \Omega := \Omega_1\left(\bigoplus_{r=1}^{\nicefrac{R}{q}}H_r\right)
         =\sum_{r=1}^{\nicefrac{R}{q}} \Delta \sigma(H_r).
     \end{align}
     An upper bound to the size of the frequency spectrum is given by
     \begin{align}
         |\Omega| \leq |\Delta \sigma(H_r)|^{\nicefrac{R}{q}}
         \leq \left(2\binom{k}{2}+1 \right)^{\nicefrac{R}{q}}
         = \left(4^q - 2^q + 1\right)^{\nicefrac{R}{q}},
     \end{align}
     where equality holds if and only if $\sigma(H_r)$ is a Golomb ruler of order $k$ for all $r=1,...,\nicefrac{R}{q}$ and each sum $\sum_{r=1}^{\nicefrac{R}{q}} \omega_r \in \sum_{r=1}^{\nicefrac{R}{q}} \Delta \sigma(H_r)$ is unique. 
     
     We explicitly construct a set of generators $H_r$ such that these conditions are satisfied. Let $H \in \End{\B^{\otimes q}}$ such that $\sigma(H)\subseteq \Z$ is a Golomb ruler. Choose any appropriate $\beta$ and let $H_r:= \beta_r \cdot H$ be as in the Golomb encoding. Then all $\sigma(H_r)$ are Golomb rulers since $\sigma(H_r) = \beta_r \cdot \sigma(H)$ and Golomb rulers are independent from non zero scaling. Now assume $\sum_{r=1}^{\nicefrac{R}{q}} \beta_r\omega_r  = \sum_{r=1}^{\nicefrac{R}{q}} \beta_r\omega_r'$
     with $w_r, w_r' \in \Delta \sigma(H)$ for all $r=1,...,\nicefrac{R}{q}$. Re-arranging the sums yields
     \begin{align*}
         \beta \cdot \left(\sum_{r=2}^{\nicefrac{R}{q}} \beta^{r-2}(\omega_r-\omega_r')\right) = \omega_1' - \omega_1.
     \end{align*}
     Since $|\omega_1' - \omega_1|\leq 2 \ell (\sigma(H))$, $\beta > 2 \ell(\sigma(H))$ and all summands are integers, this is only possible if $\omega_1 = \omega_1'$ and $\sum_{r=2}^{\nicefrac{R}{q}} \beta^{r-2}(\omega_r-\omega_r') = 0$. Repeating this argument inductively, this yields $w_r = w_r'$ for all $r=1,...,\nicefrac{R}{q}$. Hence the sums are unique and therefore the $H_r$ build a set of generators with maximal frequency spectrum as mentioned.
\end{proof}
\subsection{Relaxed Turnpike Encoding}
Next, we prove Theorem \ref{theorem: algorithmic solution relaxed turnpike problem}.
\begin{proof}
    Let $S =\lset{s_1,...,s_d} \subseteq \Z$ with $s_1 < ...<s_d$ be a solution of the relaxed turnpike problem. Given that $\Delta S = \Delta S'$ for $S' = \{0, s_2 - s_1, \ldots, s_d - s_1\}$, we may assume without loss of generality that $s_1 = 0$ and $s_i \geq 0$ for all $i \in \{1, \ldots, d\}$.  Suppose there exists an index $j \in \{1, \ldots, d-1\}$ for which $s_{j+1} - s_j > \binom{d}{2}$.  We then partition $S$ into two subsets:
    \begin{align*}
        A &:= \lset{s_1,...,s_j} \\
        B &:= \lset{s_{j+1},...,s_d}.
    \end{align*}
    Let $K:=K(S)$.
    For every element $e \in A \Delta B$ it holds
    \begin{align}
        |e| \geq s_{j+1} - s_j > \binom{d}{2} \geq K.
    \end{align}
    The assumption
    \begin{align}
        k \in \Delta S = \Delta A \cup \Delta B \cup A \Delta B
    \end{align}
    implies that $k \in \Delta A$ or $k \in \Delta B$ for all $k \in \{1, \ldots, K\}$. Define $c:= s_{j+1}-s_j - K - 1 \geq 0$ and set
    \begin{align*}
        B' &= \lset{s_{j+1}',...,s_d'}:= B-c = \lset{s_{j+1} - c,...,s_d-c}, \\
        S' &:= A \cup B'.
    \end{align*}
    Because $\Delta B' = \Delta B$, the set $\Delta S'$ also contains all integers $k=1,\ldots,K$. However, it would also contain $K+1=s_{j+1}'-s_j$, contradicting the premise that $S$ is a solution to the relaxed turnpike problem where $K$ is maximal over all sets of size $d$. Thus, our assumption that there exists an index $j$ with $s_{j+1}-s_j > \binom{d}{2}$ must be false. Consequently, the solution $S$ to the turnpike problem must be an element of the candidate set $C$.
\end{proof}
We give a proof of Theorem \ref{thm: maximality with turnpike}.
\begin{proof}
    For any univariate single layer QNN with one arbitrary generator $H\in \End{\BR}$, the frequency spectrum is given by 
    \begin{align*}
        \Omega = \Delta \sigma(H).
    \end{align*}
    By definition $K$ is maximal if and only if $\sigma(H)$ is a solution of the turnpike problem (cases in which eigenvalues occur twice can be excluded directly).
\end{proof}
Finally, we can prove Theorem \ref{thm: maximality with turnpike arbitrary k and L}
\begin{proof}
    Again, the frequency spectrum $\Omega$ is given by
    \begin{align}
        \Omega = \sum_{r=1}^{\nicefrac{R}{q}} \beta_r \Delta \sigma(H).
    \end{align}
    Each $k \in \lset{0,\ldots, \frac{(2K+1)^{\nicefrac{R}{q}}-1}{2}}$ can be (uniquely) written as
    $k=\sum_{r=1}^{\nicefrac{R}{q}} k_r (2K+1)^{r-1}$ with $k_r \in \lset{-K,\ldots,K} \subseteq \Delta \sigma(H)$, thus $k \in \Omega$.
\end{proof}
\end{appendices}


\bibliography{references}

\end{document}